\documentclass[runningheads]{llncs}

\usepackage{algorithm}
\usepackage[noend]{algorithmic}
\usepackage{amsmath}
\usepackage{amsfonts}
\usepackage{amssymb}
\usepackage{color}
\usepackage{comment}
\usepackage{enumitem}
\usepackage{graphicx}
\usepackage[bookmarks=false]{hyperref}

\usepackage[utf8]{inputenc}
\usepackage{listings}
\usepackage{makeidx}  
\usepackage[square,sort,comma,numbers]{natbib}
\usepackage{subcaption}
\usepackage{url}
\usepackage{wrapfig}
\usepackage{xcolor}
\usepackage[pdf]{graphviz}
\usepackage{multirow}
\usepackage{refcount}
\usepackage{pbox}
\usepackage{xspace}
\usepackage{bbm}
\usepackage{booktabs}
\usepackage{stmaryrd}
\usepackage[firstpage]{draftwatermark}

\usepackage[graphicx]{realboxes}
\usepackage{tikz}
\tikzset{elliptic state/.style={draw,ellipse}}
\usetikzlibrary{shapes,shapes.geometric,arrows,fit,calc,positioning,automata,}
\usetikzlibrary{automata,arrows,decorations.pathmorphing,decorations.markings,positioning,shapes,shapes.geometric,arrows.meta,bending,fit,calc}
\usepackage{tikz-qtree}
\usepackage{float}
\usepackage{floatflt}

\newtheorem{assumption}[theorem]{Assumption}

\newcommand{\s}{\mathbf{s}}

\newcommand{\code}[1]{{\texttt {#1}}}
\newcommand{\ol}[1]{\overline{#1}}

\newcommand{\p}{\phi}

\newcommand{\R}{\mathbb{R}}
\newcommand{\E}{\mathbb{E}}

\newcommand{\A}{\mathcal{A}}
\newcommand{\D}{\mathcal{D}}
\newcommand{\B}{\mathcal{B}}

\newcommand{\G}{\mathcal{G}}
\newcommand{\M}{\mathcal{M}}

\newcommand{\Rc}{\mathcal{R}}

\renewcommand{\S}{\mathcal{S}}
\renewcommand{\P}{\mathcal{P}}

\renewcommand{\L}{\mathcal{L}}
\newcommand{\Zc}{\mathcal{Z}}
\newcommand{\traj}{\mathcal{Z}}

\newcommand{\bestr}{\operatorname{br}}

\newcommand{\true}{\code{true}}
\newcommand{\false}{\code{false}}

\newcommand{\choice}[2]{#1 \; \code{or} \; #2}
\newcommand{\always}[1]{~ \code{ensuring} \; #1}
\newcommand{\eventually}[1]{\code{achieve} \; #1}
\newcommand{\semantics}[1]{{\llbracket #1 \rrbracket}}
\newcommand{\safe}{\text{safe}}

\newcommand{\prioritizedpolicies}{\mathsf{PrioritizedPolicies}}

\newcommand{\progress}{\mathsf{progress}}

\newcommand{\AbstractGraph}{\mathsf{AbstractGraph}}

\newcommand{\estimatecost}{\mathsf{EstimatePolicyWelfare}}
\newcommand{\learnpolicy}{\mathsf{LearnEdgePolicy}}
\newcommand{\MakeProduct}{\mathsf{ProductAbstractGraph}}
\newcommand{\outgoing}{\mathsf{outgoing}}
\newcommand{\pathcostmaxheap}{\mathsf{pathPolicyWelfareMaxHeap}}
\newcommand{\pop}{\mathsf{pop}}
\newcommand{\reachdistribution}{\mathsf{ReachDistribution}}
\newcommand{\StateDistribution
}{\mathsf{AverageDistribution}}
\newcommand{\TopoSort}{\mathsf{TopologicalSort}}
\newcommand{\TSortList}{\mathsf{topoSortedVertexStack}}
\newcommand{\first}{\mathsf{First}}

\newcommand{\reward}{\tilde{R}}

\newcommand{\spectrl}{\textsc{Spectrl}\xspace}

\newcommand{\nvi}{\textsc{nvi}\xspace}
\newcommand{\maqrm}{\textsc{maqrm}\xspace}
\newcommand{\nash}{\textsc{HighNashSearch}\xspace}

\begin{document}

\title{Specification-Guided Learning of Nash Equilibria with High Social Welfare\thanks{This is the extended version of a paper with the same title that appeared at CAV 2022.}
}

%
%
\author{Kishor Jothimurugan, Suguman Bansal, Osbert Bastani, Rajeev Alur}
%
%
\institute{University of Pennsylvania}

\maketitle         

\SetWatermarkAngle{0}
\SetWatermarkText{\raisebox{22.5cm}{%
  \hspace{14.62977cm}%
  \includegraphics{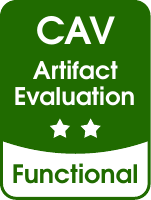}%
}}

\begin{abstract}
Reinforcement learning has been shown to be an effective strategy for automatically training policies for challenging control problems.
Focusing on non-cooperative multi-agent systems, we propose a novel reinforcement learning framework for training joint policies that form a Nash equilibrium.
In our approach, rather than providing low-level reward functions, the user provides high-level specifications that encode the objective of each agent.
Then, guided by the structure of the specifications, our algorithm searches over policies to identify one that provably forms an $\epsilon$-Nash equilibrium (with high probability).
Importantly, it prioritizes policies in a way that maximizes social welfare across all agents.
Our empirical evaluation demonstrates that our algorithm computes equilibrium policies with high social welfare, whereas state-of-the-art baselines either fail to compute Nash equilibria or compute ones with comparatively lower social welfare.
\end{abstract}

\section{Introduction}
\label{Sec:Intro}

Reinforcement learning (RL) is an effective strategy for automatically synthesizing controllers for challenging control problems. As a consequence, there has been interest in applying RL to multi-agent systems. For example, RL has been used to coordinate agents in cooperative systems to accomplish a shared goal~\cite{neary2021reward}. Our focus is on non-cooperative systems, where the agents are trying to achieve their own goals~\cite{kearns2000fast}; for such systems, the goal is typically to learn a policy for each agent such that the joint strategy forms a Nash equilibrium.

A key challenge facing existing approaches is how tasks are specified. First, they typically require that the task for each agent is specified as a reward function. However, reward functions tend to be very low-level, making them difficult to manually design; furthermore, they often obfuscate high-level structure in the problem known to make RL more efficient in the single-agent~\cite{icarte2018using} and cooperative~\cite{neary2021reward} settings. Second, they typically focus on computing an arbitrary Nash equilibrium. However, in many settings, the user is a social planner trying to optimize the overall social welfare of the system, and most existing approaches are not designed to optimize social welfare.



We propose a novel multi-agent RL framework for learning policies from high-level specifications (one specification per agent) such that the resulting joint policy (i) has high social welfare, and (ii) is an $\epsilon$-Nash equilibrium (for a given $\epsilon$). We formulate this problem as a constrained optimization problem where the goal is to maximize social welfare under the constraint that the joint policy is an $\epsilon$-Nash equilibrium. 

Our algorithm for solving this optimization problem uses an enumerative search strategy. First, it enumerates candidate policies in decreasing order of social welfare. To ensure a tractable search space, it restricts to policies that conform to the structure of the user-provided specification. Then, for each candidate policy, it
uses an explore-then-exploit self-play RL algorithm~\cite{pmlr-v119-bai20a} to compute \emph{punishment strategies} that are triggered when some agent deviates from the original joint policy. It also computes the maximum benefit each agent derives from deviating, which can be used to determine whether the joint policy augmented with punishment strategies forms an $\epsilon$-Nash equilibrium; if so, it returns the joint policy.

Intuitively, the enumerative search tries to optimize social welfare, whereas the self-play RL algorithm checks whether the $\epsilon$-Nash equilibrium constraint holds. Since this RL algorithm comes with PAC (Probably Approximately Correct) guarantees, our algorithm is guaranteed to return an $\epsilon$-Nash equilibrium with high probability.  {In summary, our contributions are as follows.

\begin{itemize}
    \item We study the problem of maximizing social welfare under the constraint that the policies form an $\epsilon$-NE. To the best of our knowledge, this problem has not been studied before in the context of learning (beyond single-step games).
    \item We provide an enumerate-and-verify framework for solving the said problem.
    \item We propose a verification algorithm with a probabilistic soundness guarantee in the RL setting of probabilistic systems with unknown transition probabilities.
\end{itemize}}


\paragraph{Motivating example.}
\begin{figure}[t]
    \centering
    \includegraphics[width=0.4\linewidth]{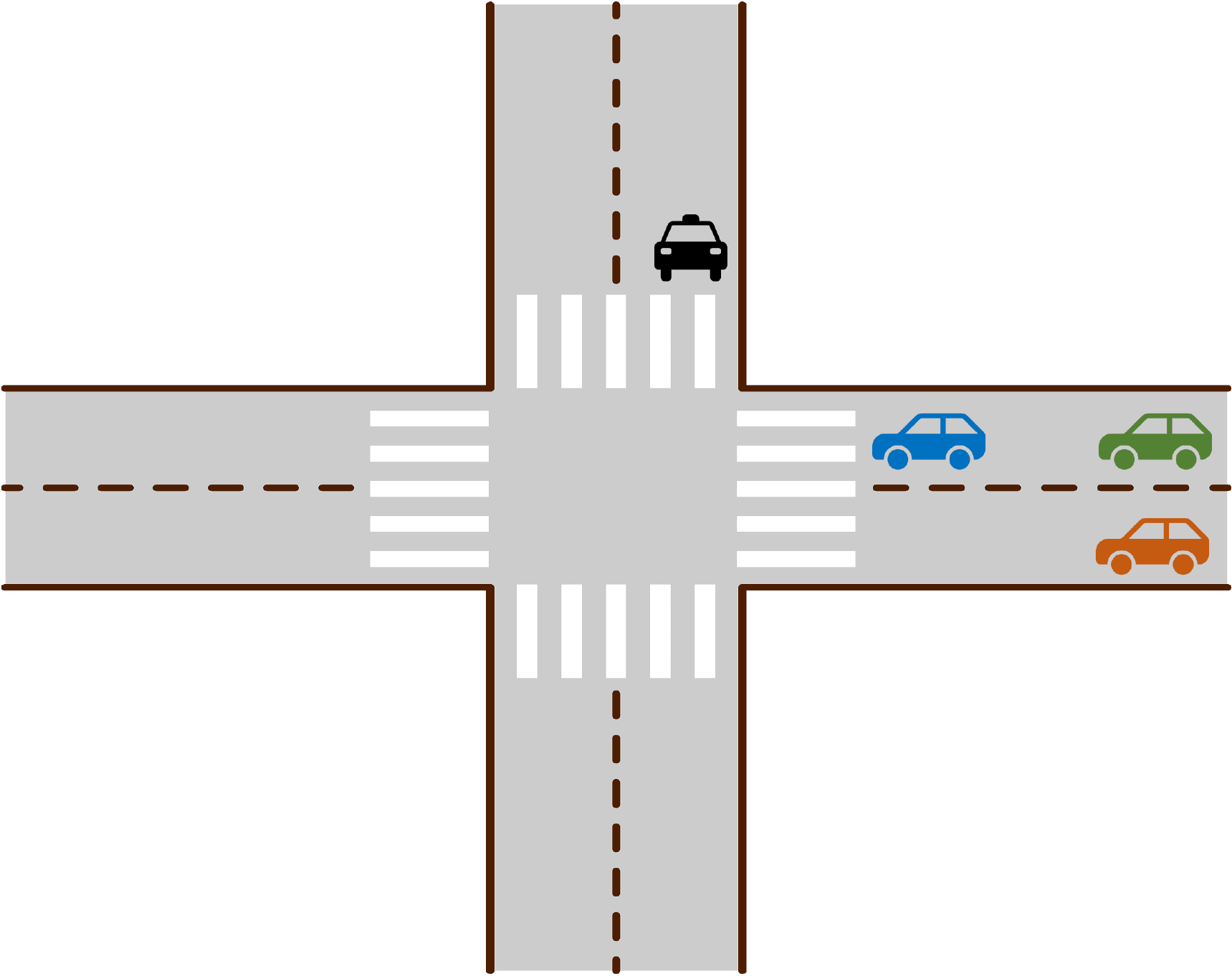}
    \caption{Intersection Example}
    \label{fig:intersection}
\end{figure}
Consider the road intersection scenario in Figure~\ref{fig:intersection}. There are four cars; three are traveling east to west and one is traveling north to south. At any stage, each car can either move forward one step or stay in place. Suppose each car's specification is as follows:
\begin{itemize}[itemsep=0pt]
\item \emph{Black car:} Cross the intersection before the green and orange cars.
\item \emph{Blue car:} Cross the intersection before the black car and stay a car length ahead of the green and orange cars.
\item \emph{Green car:} Cross the intersection before the black car.
\item \emph{Orange car:} Cross the intersection before the black car.
\end{itemize}
We also require that the cars do not crash into one another.

Clearly, not all agents can achieve their goals. The next highest social welfare is for three agents to achieve their goals. In particular, one possibility is that all cars except the black car achieve their goals. However, the corresponding joint policy requires that the black car does not move, which is not a Nash equilibrium---there is always a gap between the blue car and the other two cars behind, so the black car can deviate by inserting itself into the gap to achieve its own goal. Our algorithm uses self-play RL to optimize the policy for the black car, and finds that the other agents cannot prevent the black car from improving its outcome in this way. Thus, it correctly rejects this joint policy. Eventually, our algorithm computes a Nash equilibrium in which the black and blue cars achieve their goals.


\subsection{Related Work} 

\paragraph{Multi-agent RL.} There has been work on learning Nash equilibria in the multi-agent RL setting \cite{hu2003nash, hu1998multiagent, littman2001friend, pmlr-v54-perolat17a, prasad2015two, akchurina2008multi}; however, these approaches focus on learning an arbitrary equilibrium and do not optimize social welfare. There has also been work on studying weaker notions of equilibria in this context \cite{zinkevich2006cyclic, greenwald2003correlated}, as well as work on learning Nash equilibria in two agent zero-sum games \cite{pmlr-v119-bai20a, wei2017online, littman1994markov}.

\paragraph{RL from high-level specifications.} There has been recent work on using specifications based on temporal logic for specifying RL tasks in the single agent setting;
a comprehensive survey may be found in~\cite{alur2021framework}. There has also been recent work on using temporal logic specifications for multi-agent RL \cite{hammond2021, neary2021reward}, but these approaches focus on cooperative scenarios in which there is a common objective that all agents are trying to achieve.

\paragraph{Equilibrium  in Markov games.} There has been work on computing Nash equilibrium in Markov games~\cite{kearns2000fast, shapley1953stochastic}, including work on computing $\epsilon$-Nash equilibria from logical specifications~\cite{chatterjee2004nash,chatterjee2005two}, as well as recent work focusing on computing welfare-optimizing Nash equilibria from temporal specifications~\cite{kwiatkowska2019equilibria,kwiatkowska2020prism}; however, all these works focus on the planning setting where the transition probabilities are known. 
Checking for existence of Nash equilibrium, even in deterministic games, has been shown to be NP-complete for  reachability objectives~\cite{bouyer2010nash}.

\paragraph{Social welfare.} There has been work on computing welfare maximizing Nash equilibria for bimatrix games, which are two-player one-step Markov games with known transitions \cite{czumaj2015approximate, best-nash}; in contrast, we study this problem in the context of general Markov games.

\section{Preliminaries}
\label{Sec:problem}

\subsection{Markov Game}

\sloppy
We consider an $n$-agent Markov game $\M = (\S, \A, P, H, s_0)$ with a finite set of states $\S$, actions $\A = A_1\times\cdots\times A_n$  where $A_i$ is a finite set of actions available to agent $i$, transition probabilities $P(s'\mid s,a)$ for $s,s'\in\S$ and $a\in\A$, finite horizon $H$, and initial state $s_0$~\cite{littman1994markov}. A \emph{trajectory} $\zeta\in\traj=(\S\times\A)^*\times \S$ is a finite sequence $\zeta=s_0\xrightarrow{a_0}s_1\xrightarrow{a_1}\cdots\xrightarrow{a_{t-1}}s_t$ where $s_k \in \S$, ${a_k} \in \A$; we use $|\zeta| = t$ to denote the length of the trajectory $\zeta$ and $a_k^i\in A_i$ to denote the action of agent $i$ in $a_k$.

For any $i\in[n]$, let $\D(A_i)$ denote the set of distributions over $A_i$---i.e., $\D(A_i) = \{\Delta:A_i\to[0,1]\mid \sum_{a_i\in A_i}\Delta(a_i) = 1\}$. A \emph{policy} for agent $i$ is a function $\pi_i:\traj\to \D(A_i)$ mapping trajectories to distributions over actions. 
A policy $\pi_i$ is \emph{deterministic} if for every $\zeta\in\traj$, there is an action $a_i\in A_i$ such that $\pi_i(\zeta)(a_i) = 1$; in this case, we also use $\pi_i(\zeta)$ to denote the action $a_i$. A {\em joint policy}  $\pi: \traj\to \D(A)$ maps finite trajectories to distributions over joint actions. We use $(\pi_1,\dots, \pi_n)$ to denote the joint policy in which agent $i$ chooses its action in accordance to $\pi_i$. We denote by $\D_{\pi}$ the distribution over $H$-length trajectories in $\M$ induced by $\pi$.

We consider the reinforcement learning setting in which we do not know the probabilities $P$ but instead only have access to a simulator of $\M$. Typically, we can only sample trajectories of $\M$ starting at $s_0$. 
Some parts of our algorithm are based on an assumption which allows us to obtain sample trajectories starting at any state that has been observed before. For example, if taking action $a_0$ in $s_0$ leads to a state $s_1$, we assume we can obtain future samples starting at $s_1$.
\begin{assumption}\label{assump:model}
We can obtain samples from $P(\cdot\mid s,a)$ for any previously observed state $s$ and any action $a$.
\end{assumption}

\subsection{Specification Language}

We consider the specification language \spectrl to express  agent specifications. We choose \spectrl since there is existing work on leveraging the structure of \spectrl specifications for single-agent RL~\cite{jothimurugan2021compositional}. However, we believe our algorithm can be adapted to other specification languages as well.

Formally, a \spectrl specification is defined over a set of \emph{atomic predicates} ${\P}_0$, where every $p \in {\P}_0$ is associated with a function $\semantics{p}:\S\to\mathbb{B}=\{\true, \false\}$; we say a state $s$ \emph{satisfies} $p$ (denoted $s\models p$) if and only if $\semantics{p}(s)=\true$. 
%
The set of \emph{predicates} $\mathcal{P}$ consists of conjunctions and disjunctions of atomic predicates. The syntax of a predicate $b\in\mathcal{P}$ is given by the grammar
$
b ~::=~ p \mid (b_1 \wedge b_2) \mid (b_1 \vee b_2),
$
where $p\in\mathcal{P}_0$. Similar to atomic predicates, each predicate $b\in\mathcal{P}$ corresponds to a function $\semantics{b}:\S\to\mathbb{B}$ defined naturally over Boolean logic. 
Finally, the syntax of \spectrl is given by
\footnote{Here, \code{achieve} and \code{ensuring} correspond to the ``eventually'' and ``always'' operators in temporal logic.}
\begin{align*}
\p ~::=~ \eventually{b} \mid \p_1 \always{b} \mid \p_1; \p_2 \mid \choice{\p_1}{\p_2},
\end{align*}
where $b\in\mathcal{P}$. Each specification $\phi$ corresponds to a function $\semantics{\phi}:\traj\to\mathbb{B}$, and we say $\zeta\in\traj$ satisfies $\phi$ (denoted $\zeta\models\phi$) if and only if $\semantics{\phi}(\zeta)=\true$.
Letting $\zeta$ be a finite trajectory of length $t$, this function is defined by
\begin{align*}
\zeta&\models\eventually{b} ~&&\text{if}~ \exists\ i \leq t,~s_i\models b \\
\zeta&\models \p \always{b} ~&&\text{if}~ \zeta\models\p ~\text{and}~ \forall\ i\leq t, ~ s_i\models b \\
\zeta&\models\p_1; \p_2 ~&&\text{if}~ \exists\ i < t, ~\zeta_{0:i}\models \p_1 ~\text{and}~ \zeta_{i+1:t}\models\p_2 \\
\zeta&\models\choice{\p_1}{\p_2} ~&&\text{if}~ \zeta\models\p_1 ~\text{or}~ \zeta\models\p_2.
\end{align*}
Intuitively, the first clause means that the trajectory should eventually reach a state that satisfies the predicate $b$. The second clause says that the trajectory should satisfy specification $\p$ while always staying in states that satisfy $b$. The third clause says that the trajectory should sequentially satisfy $\p_1$ followed by $\p_2$. The fourth clause means that the trajectory should satisfy either $\p_1$ or $\p_2$.

\subsection{Abstract Graphs}
\spectrl specifications can be represented by {\em abstract graphs} {which are DAG-like structures in which each vertex represents a set of states (called subgoal regions) and each edge represents a set of concrete trajectories that can be used to transition from the source vertex to the target vertex without violating safety constraints.} 

\begin{definition}
\rm
An {\em abstract graph} $\G = (U,E,u_0,F,\beta,\traj_{\safe})$ is a directed acyclic graph (DAG) with vertices $U$,
(directed) edges $E\subseteq U\times U$, initial vertex $u_0\in U$, final vertices $F\subseteq U$, subgoal region map $\beta:U\to2^S$ such that for each $u\in U$, $\beta(u)$ is a subgoal region,\footnote{We do not require that the subgoal regions partition the state space or that they be non-overlapping.} and \emph{safe trajectories}
$
\traj_\safe = \bigcup_{e \in E}\traj_\safe^e\cup\bigcup_{f \in F}\traj_\safe^f,
$
where $\traj_\safe^e\subseteq\traj$ denotes the safe trajectories for edge $e \in E$ and $\traj_\safe^f\subseteq\traj$ denotes the safe trajectories for final vertex $f\in F$.
\end{definition}
Intuitively, $(U,E)$ is a standard DAG, and $u_0$ and $F$ define a graph reachability problem for $(U,E)$. Furthermore, $\beta$ and $\traj_{\safe}$ connect $(U,E)$ back to the original MDP $\M$; in particular, for an edge $e=u\to u'$, $\traj_{\safe}^e$ is the set of safe trajectories in $\M$ that can be used to transition from $\beta(u)$ to $\beta(u')$. 
\begin{definition}
\rm
A trajectory $\zeta=s_0\xrightarrow{a_0}s_1\xrightarrow{a_1}\cdots\xrightarrow{a_{t-1}}s_t$ in $\M$ satisfies the abstract graph $\G$ (denoted $\zeta\models \G$) if there is a sequence of indices $0=k_0\leq k_1<\cdots<k_\ell\leq t$ and a path $\rho=u_0\to u_1\to\cdots\to u_\ell$ in $\G$ such that
\begin{itemize}[topsep=0pt,itemsep=0ex,partopsep=1ex,parsep=1ex]
\item $u_\ell\in F$,
\item for all $z\in\{0,\ldots,\ell\}$, we have $s_{k_z}\in \beta(u_z)$,
\item for all $z < \ell$, letting $e_z=u_z\to u_{z+1}$, we have $\zeta_{k_z:k_{z+1}}\in\traj_{\safe}^{e_z}$, and
\item $\zeta_{k_\ell:t}\in\traj_\safe^{u_\ell}$.
\end{itemize}
\end{definition}
The first two conditions state that the trajectory should visit a sequence of subgoal regions corresponding to a path from the initial vertex to some final vertex, and the last two conditions state that the trajectory should be composed of subtrajectories that are safe according to $\traj_\safe$.



Prior work shows that for every  \spectrl specification $\p$, we can construct an abstract graph $\G_\p$ such that for every trajectory $\zeta\in\traj$, $\zeta\models\p$ if and only if $\zeta\models\G_\p$~\cite{jothimurugan2021compositional}. Finally, the number of states in the abstract graph is linear in the size of the specification. 


\subsection{Nash Equilibrium and Social Welfare}

Given a Markov game $\M$ with unknown transitions and \spectrl specifications $\p_1,\dots, \p_n$ for the $n$ agents respectively, the score of agent $i$ from a joint policy $\pi$ is given by 
$$J_i(\pi) = \Pr_{\zeta\sim\D_{\pi}}[\zeta\models \p_i].$$
 
Our goal is to compute a {\em high-value} $\epsilon$-Nash equilibrium in $\M$ w.r.t these scores. Given a joint policy $\pi = (\pi_1,\dots, \pi_n)$ and an alternate policy $\pi'_i$ for agent $i$, let $(\pi_{-i},\pi'_i)$ denote the joint policy $ (\pi_1, \dots , \pi'_i, \dots, \pi_n)$. Then, a joint policy $\pi$ is an \emph{$\epsilon$-Nash equilibrium} if for all agents $i$ and all alternate policies $\pi'_i$, $J_i(\pi)\geq J_i((\pi_{-i},\pi'_i)) - \epsilon$. Our goal is to compute a joint policy $\pi$ that maximizes the social welfare given by
$$\operatorname*{\code{welfare}}(\pi) = \frac{1}{n}\sum_{i=1}^n J_i(\pi)$$ subject to the constraint that
$\pi$ is an $\epsilon$-Nash equilibrium.

\section{Overview}

Our framework for computing a high-welfare $\epsilon$-Nash equilibrium consists of two phases. The first phase is a {\em prioritized enumeration} procedure that learns deterministic joint policies in the environment and ranks them in decreasing order of social welfare. The second phase is a {\em verification phase} that checks whether a given joint policy can be extended to an $\epsilon$-Nash equilibrium by adding punishment strategies. A policy is returned if it passes the verification check in the second phase. Algorithm~\ref{Alg:Overall} summarizes our framework.

For the enumeration phase, it is impractical to enumerate all joint policies even for small environments, since the total number of deterministic joint policies is $\Omega(|\A|^{|\S|^{H-1}})$, which is $\Omega(2^{n|\S|^{H-1}})$ if each agent has atleast two actions. Thus, in the prioritized enumeration phase, we apply a specification-guided heuristic to reduce the number of joint policies considered.
The resulting search space is independent of $|\S|$ and $H$, depending only on the specifications $\{\p_i\}_{i\in[n]}$.
Since the transition probabilities are unknown, these joint policies are trained using an efficient compositional RL approach.

Since the joint policies are trained cooperatively, they are typically not $\epsilon$-Nash equilibria. Hence, in the verification phase, we use a probably approximately correct (PAC) procedure (Algorithm~\ref{alg:NEverification}) to determine whether a given joint policy can be modified by adding \emph{punishment strategies} to form an $\epsilon$-Nash equilibrium. Our approach is to reduce this problem to solving two-agent zero-sum games. The key insight is that for a given joint policy to be an $\epsilon$-Nash equilibrium, unilateral deviations by any agent must be successfully punished by the coalition of all other agents. In such a \emph{punishment game}, the deviating agent attempts to maximize its score while the coalition of other agents attempts to minimize its score, leading to a competitive min-max game between the agent and the coalition. If the deviating agent can improve its score by a margin $\ge\epsilon$, then the joint policy cannot be extended to an $\epsilon$-Nash equilibrium. Alternatively, if no agent can increase its score by a margin $\ge\epsilon$, then the joint policy (augmented with punishment strategies) is an $\epsilon$-Nash equilibrium. Thus, checking if a joint policy can be converted to an $\epsilon$-Nash equilibrium reduces to solving a two-agent zero-sum game for each agent.
Each punishment game is solved using a self-play RL algorithm for learning policies in min-max games with unknown transitions~\cite{pmlr-v119-bai20a}, after converting specification-based scores to  reward-based scores. While the initial joint policy is deterministic, the punishment strategies can be probabilistic. 

Overall, we provide the guarantee that with high probability, if our algorithm returns a joint policy, it will be an $\epsilon$-Nash equilibrium.

\begin{algorithm}[t]
\begin{algorithmic}[1]
\STATE $\prioritizedpolicies \gets \textsc{PrioritizedEnumeration}(\M,\p_1,\dots,\p_n)$ \label{algline:overview:policyenumeration1}
\FOR{joint policy $\pi \in \prioritizedpolicies$}\label{algline:overview:forpolicyenumeration}
    \STATE{\color{blue} // Can $\pi$ be extended to an $\epsilon$-NE?}
    \STATE $\code{isNash}, \tau \gets \textsc{VerifyNash}(\M, \pi, \p_1,\cdots,\p_n,\epsilon, \delta, p)$ \label{algline:verification-begin}
    \STATE \textbf{if} $\code{isNash}$ \textbf{then} \textbf{return} $\pi\Join\tau$ {\color{blue} // Add punishment strategies}\label{algline:verification-end}
\ENDFOR 
\STATE{\textbf{return} No $\epsilon$-NE found}
\caption{\nash\\
\textbf{Inputs:} Markov game (with unknown transition probabilities) $\M$ with $n$-agents, agent specifications $\p_1, \dots,\p_n$, Nash factor $\epsilon$, precision $\delta$, failure probability $p$.\\
\textbf{Outputs:} $\epsilon$-NE, if found.}
\label{Alg:Overall}
\end{algorithmic}
\end{algorithm}

\sloppy

\section{Prioritized Enumeration}
\label{Sec:search}

We summarize our specification-guided compositional RL algorithm for learning a finite number of deterministic joint policies in an unknown environment under Assumption~\ref{assump:model}; details are in Appendix~\ref{app:search}. 
These policies are then ranked in decreasing order of their (estimated) social welfare.

Our learning algorithm harnesses the structure of specifications, exposed by their abstract graphs, to curb the number of joint policies to learn. 
For every set of {\em active agents} $B \subseteq [n]$, we construct a product abstract graph, from the abstract graphs of {all} active agents' specifications. A property of this product is that if a trajectory $\zeta$ in $\M$ corresponds to a path in the product that ends in a final state then $\zeta$ satisfies the specification of \emph{all} active agents. Then, our procedure learns one joint policy for every path in the product graph that reaches a final state. {Intuitively, policies learned using the product graph corresponding to a set of active agents $B$ aim to maximize satisfaction probabilities of all agents in $B$.} By learning joint policies for every set of active agents, we are able to learn policies under which some agents may not satisfy their specifications. This enables learning joint policies in non-cooperative settings. Note that the number of paths (and hence the number of policies considered) is independent of $|\S|$ and $H$, and depends only on the number of agents and their specifications.  

One caveat is that the number of paths may be exponential in the number of states in the product graph. It would be impractical to na\"{i}vely learn a joint policy for every path. Instead, we design an efficient compositional RL algorithm that learns a joint policy for each edge in the product graph; these edge policies are then composed together to obtain joint policies for paths in the product graph. 

\subsection{Product Abstract Graph}

\begin{figure}[t]
	\centering
	\begin{minipage}{0.40\textwidth}
		\centering
		\begin{tikzpicture}[shorten >=1pt,node distance=3.5cm,on grid,auto] 
		
		\node[state] (q_0)   { \footnotesize $u_1$}; 
        \node[state, accepting] (q_1) [right=of q_0] {\footnotesize $v_1$}; 
		
		\path[->] 
		(q_0)   edge node [above]{ \footnotesize $\mathcal{Z}_1 = \mathcal{Z}_{\mathsf{no\_collision}}$}  (q_1);
		\end{tikzpicture}
		\caption{Abstract Graph of black car.}
		\label{Fig:BlackCar}
		\bigskip
		
		\begin{tikzpicture}[shorten >=1pt,node distance=3.5cm,on grid,auto] 
		
		\node[state] (q_0)   { \footnotesize $u_2$}; 
        \node[state, accepting] (q_1) [right=of q_0] {\footnotesize $v_2$}; 
		
		\path[->] 
		(q_0)   edge node [above]{ $\mathcal{Z}_2=\mathcal{Z}_{\mathsf{no\_collision}}\cap$}    (q_1)
		edge node [below]{$\mathcal{Z}_{\mathsf{distance}\_\mathsf{{GreenOrange}}}$}    (q_1);
		\end{tikzpicture}
		\caption{Abstract Graph of blue car.}
		\label{Fig:BlueCar}
	\end{minipage}
	\hfill
	\centering
	\begin{minipage}{0.5\textwidth}
		\centering
		\begin{tikzpicture}[shorten >=1pt,node distance=2.35cm,on grid,auto] 
		\node[state] (q_0)   { \footnotesize $u_1, u_2$}; 
		\node[state] (q_1) [above right=of q_0] {\footnotesize $u_1, v_2$}; 
		\node[state] (q_2) [below right=of q_0] {\footnotesize $v_1, u_2$}; 
    	\node[state, accepting] (q_3) [right=3.5cm of q_0] {\footnotesize $v_1, v_2$}; 

		\path[->] 
		(q_0)  edge node {  $\first(\mathcal{Z}_1) \cap \mathcal{Z}_2$
		}  (q_1)
		edge node [left] {
		$\mathcal{Z}_1 \cap \first(\mathcal{Z}_2)\ $
		}  (q_2)
		edge node {
		$\mathcal{Z}_1 \cap \mathcal{Z}_2$
		}  (q_3)
		(q_1) edge node {
		$\mathcal{Z}_1 \cap \mathcal{Z}^{v_2}$
		}  (q_3)
		(q_2) edge node [right] { $\ \mathcal{Z}^{v_1} \cap \mathcal{Z}_2$
		}  (q_3);
		
		\end{tikzpicture}
		\caption{Product Abstract Graph of black and blue cars.
		$\mathcal{Z}^{v_1}$ and $\mathcal{Z}^{v_2}$ refer to safe trajectories after the black and blue cars have reached their final states, respectively.}
		\label{Fig:Product}
	\end{minipage}
\end{figure}

Let $\p_1,\dots, \p_n$ be the specifications for the $n$-agents, respectively, and let $\G_{i} = (U_i, E_i, u_0^i, F_i, \beta_i, \ol{\Zc}_{\safe,i})$ be the abstract graph of specification $\p_i$ in the environment $\M$. We construct a product abstract graph for every set of active agents in $[n]$.
The product graph for a set of active agents $B \subseteq [n]$ is used to learn joint policies which satisfy the specification of all agents in $B$ with high probability.

\begin{definition}
Given a set of agents $B  = \{i_1, \dots, i_m\} \subseteq [n]$,  the product graph $\G_B = (\ol{U}, \ol{E}, \ol{u}_0, \ol{F}, \ol{\beta}, \ol{\Zc}_{\safe}) $ is the asynchronous product of $\G_{i}$ for all $i \in B$, with
\begin{itemize}
    \item $\ol{U} = \prod_{i\in B}U_{i}$ is the set of product vertices,
    \item An edge $e = (u_{i_1},\ldots,u_{i_m})\to(v_{i_1},\ldots,v_{i_m})\in \ol{E}$ if at least for one agent $i\in B$ the edge  $u_{i}\to v_{i}\in E_{i}$ and for the remaining agents, $u_{i} = v_{i}$,
    \item $\ol{u}_0 = (u_0^{i_1},\ldots,u_0^{i_m})$ is the initial vertex,
    \item $\ol{F}= \Pi_{i\in B}F_{i}$ is the set of final vertices,
    \item $\ol{\beta} = (\beta_{i_1},\ldots,\beta_{i_m})$ is the collection of concretization maps, and
    \item $\ol{\traj}_{\safe} = (\ol{\traj}_{\safe,i_1},\ldots,\ol{\traj}_{\safe,i_m})$ is the collection of safe trajectories.
\end{itemize}

\end{definition}
 
We denote the $i$-th component of a product vertex $\ol{u}\in \ol{U}$ by $u_i$ for agent $i\in B$.
Similarly, the $i$-th component in an edge $e = \ol{u}\to\ol{v}$ is denoted by $e_i = u_i\to v_i$ for $i \in B$; note that $e_i$ can be a self loop which is not an edge in $\G_i$. For an edge $e\in \ol{E}$, we denote the set of agents $i\in B$ for which $e_i \in E_{i}$, and not a self loop,  by $\progress(e)$.

Abstract graphs of the black car and the blue car from the motivating example are shown in Figures~\ref{Fig:BlackCar} and \ref{Fig:BlueCar} respectively. The vertex $v_1$ denotes the subgoal region $\beta_{\text{black}}(v_1)$ consisting of states in which the black car has crossed the intersection but the orange and green cars have not. The subgoal region $\beta_{\text{blue}}(v_2)$ is the set of states in which the blue car has crossed the intersection but the black car has not. $\traj_1$ denotes trajectories in which the black car does not collide and $\traj_2$ denotes trajectories in which the blue car does not collide and stays a car length ahead of the orange and green cars. The product abstract graph for the set of active agents $B = \{\text{black, blue}\}$ is shown in Fig~\ref{Fig:Product}. The safe trajectories on the edges reflect the notion of \emph{achieving} a product edge which we discuss below.

A trajectory $\zeta = s_0\xrightarrow{a_0}s_1\xrightarrow{a_1}\cdots\xrightarrow{a_{t-1}}s_t$ \emph{achieves} an edge $e = \ol{u}\to \ol{v}$ in $\G_B$  if all progressing agents $i\in\progress(e)$ reach their target subgoal region $\beta_i(v_i)$ along the trajectory and the trajectory is safe for all agents in $B$. 
For a progressing agent $i\in \progress(e)$, the initial segment of the rollout until the agent reaches its subgoal region should be safe with respect to the edge $e_i$. 
After that, the rollout should be safe with respect to every future possibility for the agent. This is required to ensure continuity of the rollout into adjacent edges in the product graph $\G_B$. For the same reason, we require that the entire rollout is safe with respect to all future possibilities for non-progressing agents. Note that we are not concerned with non-active agents in $[n]\setminus B$. In order to formally define this notion, we need to setup some notation.

For a predicate $b\in\P$, let the set of safe trajectories w.r.t. $b$ be given by $\traj_b = \{\zeta = s_0\xrightarrow{a_0}s_1\xrightarrow{a_1}\cdots\xrightarrow{a_{t-1}}s_t\in \traj\mid \forall\ 0\leq k \leq t, s_k \models b \}$. It is known that safe trajectories along an edge in an abstract graph constructed from a \spectrl specification is either of the form $\traj_b$ or $\traj_{b_1} \circ \Zc_{b_2}$, where $b, b_1, b_2\in\P$ and $\circ$ denotes concatenation \cite{jothimurugan2021compositional}. In addition, for every final vertex $f$, $\traj_{\safe}^f$ is of the form $\traj_b$ for some $b\in\P$. We define $\first$ as follows: 
\begin{align*}
    \first(\Zc') & = 
    \begin{cases}
        \Zc_b,  &\text{if } \Zc' = \Zc_{b}\\
        \Zc_{b_1}, &\text{if } \Zc' = \Zc_{b_1}\circ \Zc_{b_2}
    \end{cases}
\end{align*}

We are now ready to define the notion of satisfiability of a product edge.

\begin{definition}
\label{def:edgerollout}
\rm
A rollout $\zeta = s_0\xrightarrow{a_0}s_1\xrightarrow{a_1}\cdots\xrightarrow{a_{t-1}}s_k$ \emph{achieves an edge} $e = \ol{u}\to \ol{v}$ in $\G_B$ (denoted $\zeta\models_B e$) if
\begin{enumerate}
    \item for all progressing agents $i\in \progress(e)$, there exists an index $k_i\leq k$ such that $s_{k_i}\in \beta_i(v_i)$ and $\zeta_{0:k_i} \in \traj^{e_i}_{\safe, i}$. If $v_i \in F_i$ then $\zeta_{k_i:k} \in \traj^{v_i}_{\safe, i}$.  Otherwise,  $\zeta_{k_i:k} \in \mathsf{First}(\traj^{v_i\rightarrow w_i}_{\safe,i})$ for all $ w_i \in \outgoing(v_i)$. Furthermore, we require $k_i > 0$ if $u_i\neq u_0^i$.
    
    \item for all non-progressing agents $i \in B\setminus\progress(e)$, if $u_i \notin F_i$, $\zeta\in\mathsf{First}(\traj_{\safe,i}^{u_i\rightarrow w_i})$ for all $w_i \in \outgoing(u_i)$. Otherwise (if $u_i \in F_i$), $\zeta \in \traj^{u_i}_{\safe, i}$
\end{enumerate}
 
\end{definition}

We can now define what it means for a trajectory to achieve a path in the product graph $\G_B$.
\begin{definition}
\label{def:achievepath}
\rm
Given $B\subseteq [n]$, a rollout $\zeta = s_0\to\cdots\to s_t $ {\em achieves} a path $\rho = \ol{u}_0 \to \cdots \to \ol{u}_{\ell}$ in $\G_B$ (denoted $\zeta \models_B \rho$) if there exists indices $0=k_0 \leq k_1 \leq \dots \leq k_\ell \leq t$ such that (i) $\ol{u}_{\ell}\in \ol{F}$, (ii) $\zeta_{k_z:k_{z+1}}$ achieves $\ol{u}_{z}\to \ol{u}_{z+1}$ for all $0\leq z< \ell$, and (iii) $\zeta_{k_\ell:t} \in \Zc_{\safe,i}^{u_{\ell,i}}$ for all $i\in B$.
\end{definition}

\begin{theorem}
\label{thrm:productgraph}
Let $\rho = \ol{u_0} \to \ol{u_1} \to \cdots \to \ol{u_\ell}$ be a path in the product abstract graph $\G_B$ for $B \subseteq [n]$. Suppose trajectory $\zeta \models_B \rho$. Then $\zeta \models \p_i$ ~for all $i \in B$.
\end{theorem}

That is, joint policies that maximize the probability of achieving paths in the product abstract graph $\G_B$ have high social welfare w.r.t. the active agents $B$.

\subsection{Compositional RL Algorithm}

Our compositional RL algorithm learns joint policies corresponding to paths in product abstract graphs. For every $B\subseteq [n]$, it learns a joint policy $\pi_e$ for each edge in the product abstract graph $\G_B$, which is the (deterministic) policy that maximizes the probability of achieving $e$ from a given initial state distribution.
We assume all agents are acting cooperatively; thus, we treat the agents as one and use single-agent RL to learn each edge policy. We will check whether any deviation to this co-operative behaviour by any agent can be punished by the coalition of other agents in the verification phase.
The reward function is designed to capture the reachability objective of progressing agents and the safety objective of all active agents. 

The edges are learned in topological order, allowing us to learn an induced state distribution for each product vertex $\ol{u}$ prior to learning any edge policies from $\ol{u}$; this distribution is used as the initial state distribution when learning outgoing edge policies from $\ol{u}$. In more detail, the distribution for the initial vertex of $\G_B$ is taken to be the initial state distribution of the environment; for every other product vertex, the distribution is the average over distributions induced by executing edge policies for all incoming edges. This is possible because the product graph is a DAG.

Given edge policies $\Pi$ along with a path
$
\rho=\ol{u}_{0}\to \ol{u}_{1}\to \cdots\to \ol{u}_{\ell} = \ol{u} \in \ol{F} 
$
in $\G_B$, we define a \emph{path policy} ${\pi}_{\rho}$ to navigate from $\ol{u}_{0}$ to $\ol{u}$. In particular, ${\pi}_{\rho}$ executes $\pi_{e[z]}$, where $e[z] = \ol{u}_{z}\to \ol{u}_{z+1}$ (starting from $z=0$) until the resulting trajectory achieves $e[z]$, after which it increments $z\gets z+1$ (unless $z=\ell$). That is, ${\pi}_{\rho}$ is designed to achieve the sequence of edges in $\rho$. 
Note that $\pi_\rho$ is a finite-state deterministic joint policy in which vertices on the path correspond to the memory states that keep track of the index of the current policy. This way, we obtain finite-state joint policies by learning edge policies only. 

This process is repeated for all sets of active agents $B \subseteq [n]$. These finite-state joint policies are then ranked by estimating their social welfare on several simulations. 
\section{Nash Equilibria Verification}
\label{Sec:neverification}

\begin{figure}[t]
    \centering
    \includegraphics[width=0.6\linewidth]{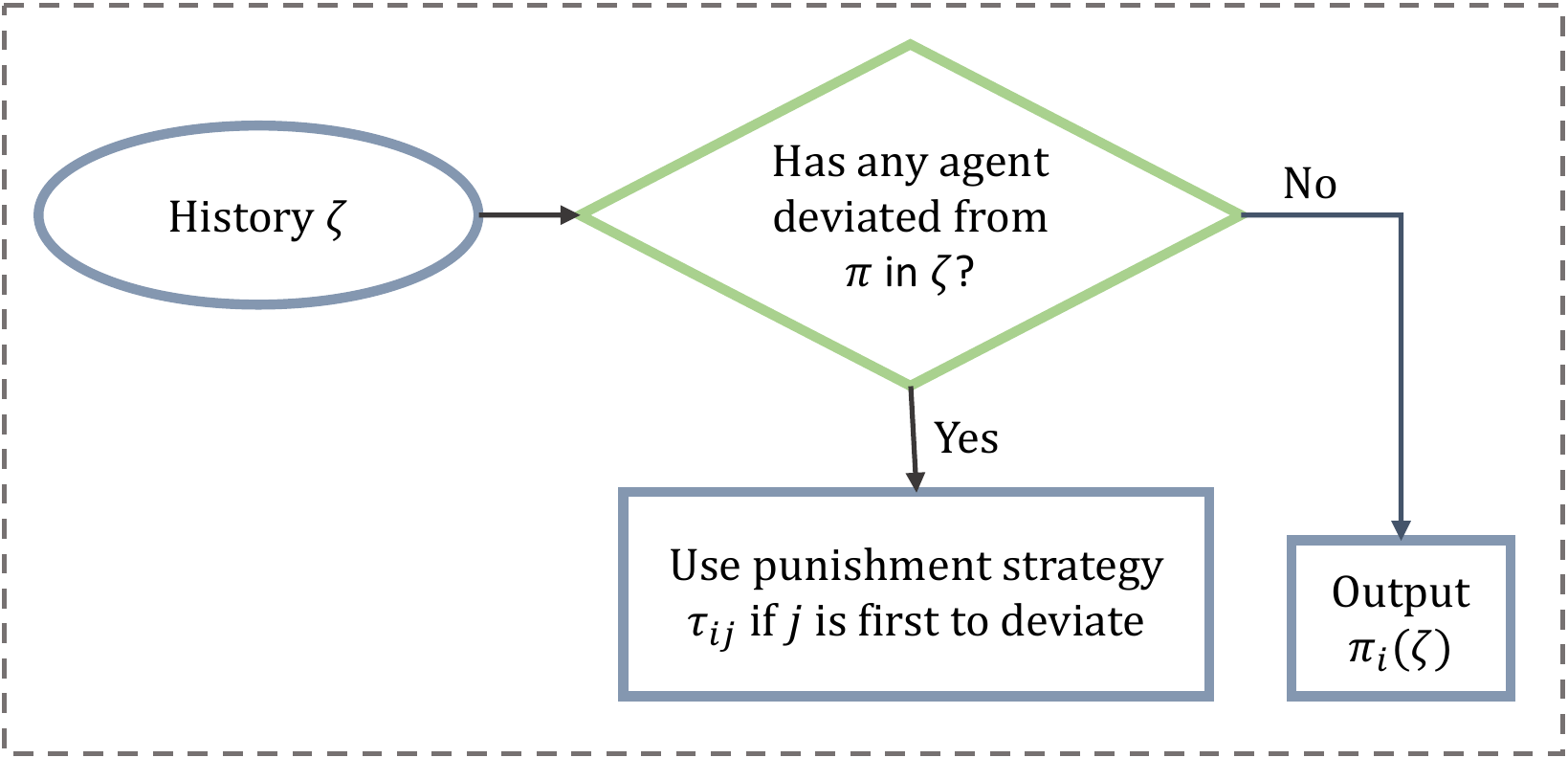}
    \caption{$\pi_i$ augmented with punishment strategies.}
    \label{fig:augmented_policy}
\end{figure}

The prioritized enumearation phase produces a list of path policies which are ranked by the total sum of scores. Each path policy is deterministic and also finite state. Since the joint policies are trained cooperatively, they are typically not $\epsilon$-Nash equilibria.
Thus, our verification algorithm not only tries to prove that a given joint policy is a $\epsilon$-Nash equilibrium, but also tries to modify it so it satisfies this property. In particular,
our verification algorithm attempts to modify a given joint policy by adding \emph{punishment strategies} so that the resulting policy is an $\epsilon$-Nash equilibrium.

Concretely, it takes as input a finite-state deterministic joint policy $\pi = (M, \alpha, \sigma, m_0)$ where $M$ is a finite set of \emph{memory states}, $\alpha: \S\times\A\times M\to M$ is the memory update function, $\sigma:\S\times M \to \A$ maps states to (joint) actions and $m_0$ is the initial policy state. The \emph{extended memory update function} $\hat{\alpha}: \traj \to M$ is given by $\hat{\alpha}(\epsilon) = m_0$ and $\hat{\alpha}(\zeta s_t a_{t}) = \alpha(s_t, a_t, \hat{\alpha}(\zeta))$. Then, $\pi$ is given by $\pi(\zeta s_t) = \sigma(s_t, \hat{\alpha}(\zeta))$. The policy $\pi_i$ of agent $i$ simply chooses the $i^{\text{th}}$ component of $\pi(\zeta)$ for any history $\zeta$.

The verification algorithm learns one punishment strategy $\tau_{ij}:\traj\to\D(A_i)$ for each pair $(i,j)$ of agents. As outlined in Figure~\ref{fig:augmented_policy}, the modified policy for agent $i$ uses $\pi_i$ if every agent $j$ has taken actions according to $\pi_j$ in the past. In case some agent $j'$ has taken an action that does not match the output of $\pi_{j'}$, then agent $i$ uses the punishment strategy $\tau_{ij}$, where $j$ is the agent that deviated the earliest (ties broken arbitrarily). The goal of verification is to check if there is a set of punishment strategies $\{\tau_{ij}\mid i\neq j\}$ such that after modifying each agent's policy to use them, the resulting joint policy is an $\epsilon$-Nash equilibrium.

\subsection{Problem Formulation}

We denote the set of all punishment strategies of agent $i$ by $\tau_i=\{\tau_{ij}\mid j\neq i\}$. We define the composition of $\pi_i$ and $\tau_i$ to be the policy $\tilde{\pi}_i = \pi_i\Join\tau_i$ such that for any trajectory $\zeta = s_0\xrightarrow{a_0}\cdots\xrightarrow{a_{t-1}}s_t$, we have
\begin{itemize}
\item $\tilde{\pi}_i(\zeta) = \pi_i(\zeta)$ if for all $0\leq k < t$, $a_k = \pi(\zeta_{0:k})$---i.e., no agent has deviated so far,
\item $\tilde{\pi}_i(\zeta) = \tau_{ij}(\zeta)$ if there is a $k$ such that (i) $a_k^j\neq \pi_j(\zeta_{0:k})$ and (ii) for all $\ell < k$,  $a_\ell = \pi(\zeta_{0:\ell})$. If there are multiple such $j$'s, an arbitrary but consistent choice is made (e.g., the smallest such $j$).
\end{itemize}
Given a finite-state deterministic joint policy $\pi$, the verification problem is to check if there exists a set of punishment strategies $\tau = \bigcup_i \tau_i$ such that the joint policy $\tilde{\pi} = \pi\Join\tau = ({\pi}_1\Join\tau_1,\ldots,{\pi}_n\Join\tau_n)$ is an $\epsilon$-Nash equilibrium. In other words, the problem is to check if there exists a policy $\tilde{\pi}_i$ for each agent $i$ such that (i) $\tilde{\pi}_i$ follows $\pi_i$ as long as no other agent $j$ deviates from $\pi_j$ and (ii) the joint policy $\tilde{\pi} = (\tilde{\pi}_1,\ldots,\tilde{\pi}_n)$ is an $\epsilon$-Nash equilibrium.

\subsection{High-Level Procedure}

Our approach is to compute the best set of punishment strategies $\tau^*$ w.r.t. $\pi$ and check if $\pi\Join\tau^*$ is an $\epsilon$-Nash equilibrium. The best punishment strategy against agent $j$ is the one that minimizes its incentive to deviate. To be precise, we define the best response of $j$ with respect to a joint policy $\pi' = (\pi_1',\ldots,\pi_n')$ to be
$\bestr_j(\pi') \in \arg\max_{\pi_j''}J_j(\pi_{-j}', \pi_j'')$.
Then, the best set of punishment strategies $\tau^*$ w.r.t. $\pi$ is one that minimizes the value of $\bestr_j(\pi\Join\tau)$ for all $j\in[n]$. To be precise, define $\tau[j] = \{\tau_{ij}\mid i\neq j\}$ to be the set of punishment strategies \emph{against} agent $j$. Then, we want to compute $\tau^*$ such that for all $j$,
\begin{equation}\label{eq:tau_objective}
    \tau^* \in \arg\min_{\tau} J_j((\pi\Join\tau)_{-j}, \bestr_j(\pi\Join\tau)).
\end{equation}
We observe that for any two sets of punishment strategies $\tau$, $\tau'$ with $\tau[j] = \tau'[j]$ and any policy $\pi_j'$, we have $J_j((\pi\Join\tau)_{-j}, \pi_j') = J_j((\pi\Join\tau')_{-j}, \pi_j')$. This is because, for any $\tau$, punishment strategies in $\tau\setminus\tau[j]$ do not affect the behaviour of the joint policy $((\pi\Join\tau)_{-j}, \pi_j')$, since no agent other than agent $j$ will deviate from $\pi$.
Hence, $\bestr_j(\pi\Join\tau)$ as well as $J_j((\pi\Join\tau)_{-j},\bestr_j(\pi\Join\tau))$ are independent of $\tau\setminus\tau[j]$; therefore, we can separately compute $\tau^*[j]$ (satisfying Equation~\ref{eq:tau_objective}) for each $j$ and take $\tau^* = \bigcup_j \tau^*[j]$. The following theorem follows from the definition of $\tau^*$; see Appendix~\ref{sec:best_puni} for a proof.
\begin{theorem}\label{thm:best_pun}
Given a finite-state deterministic joint policy $\pi=(\pi_1,\ldots,\pi_n)$, if there is a set of punishment strategies $\tau$ such that $\pi\Join\tau$ is an $\epsilon$-Nash equilibrium, then $\pi\Join\tau^*$ is an $\epsilon$-Nash equilibrium, where $\tau^*$ is the set of best punishment strategies w.r.t. $\pi$. Furthermore, $\pi\Join\tau^*$ is an $\epsilon$-Nash equilibrium iff for all $j$, $$J_j((\pi\Join\tau^*)_{-j}, \bestr_j(\pi\Join\tau^*))-\epsilon \leq J_j(\pi\Join\tau^*) = J_j(\pi).$$
\end{theorem}
Thus, to solve the verification problem, it suffices to compute (or estimate), for all $j$, the optimal deviation scores
\begin{equation}\label{eqn:min_max}
\code{dev}_j^{\pi} = \min_{\tau[j]}\max_{\pi_j'} J_j((\pi\Join\tau)_{-j}, \pi_j').
\end{equation}

\subsection{Reduction to Min-Max Games}

Next, we describe how to reduce the computation of optimal deviation scores to a standard self-play RL setting. We first translate the problem from the specification setting to a reward-based setting using \emph{reward machines}.

\paragraph{Reward Machines.} A \emph{reward machine (RM)} \cite{icarte2018using} is a tuple $\Rc = (Q, \delta_u, \delta_r, q_0)$ where $Q$ is a finite set of states, $\delta_u: \S\times\A\times Q\to Q$ is the state transition function, $\delta_r: \S\times Q\to [-1,1]$ is the reward function and $q_0$ is the initial RM state. Given a trajectory $\zeta = s_0 \xrightarrow{a_0}\ldots \xrightarrow{a_{t-1}}s_t$, the reward assigned by $\Rc$ to $\zeta$ is $\Rc(\zeta) = \sum_{k=0}^{t-1}\delta_r(s_k, q_k)$, where $q_{k+1} = \delta_u(s_{k}, a_k, q_k)$ for all $k$. For any \spectrl specification $\p$, we can construct an RM such that the reward assigned to a trajectory $\zeta$ indicates whether $\zeta$ satisfies $\p$; a proof can be found in Appendix~\ref{appendix:rm}.

\begin{theorem}\label{thm:rm_const}
Given any \spectrl specification $\p$, we can construct an RM $\Rc_\p$ such that for any trajectory $\zeta$ of length $t+1$, $\Rc_\p(\zeta) = \mathbbm{1}(\zeta_{0:t}\models\p)$.
\end{theorem}

For an agent $j$, let $\Rc_j$ denote $\Rc_{\p_j} = (Q_j, \delta_u^j, \delta_r^j, q_0^j)$. Letting $\tilde{\D}_{\pi}$ be the distribution over length $H+1$ trajectories induced by using $\pi$, we have
$\E_{\zeta\sim\tilde{\D}_{\pi}}[\Rc_j(\zeta)] = J_j(\pi).$
The deviation values defined in Eq.~\ref{eqn:min_max} are now min-max values of expected reward, except that it is not in a standard min-max setting since the policy of every non-deviating agent $i\neq j$ is constrained to be of the form $\pi_i\Join\tau_i$. This issue can be handled by considering a product of $\M$ with the reward machine $\Rc_j$ and the finite-state joint policy $\pi$. The following theorem follows naturally; details are in Appendix~\ref{app:puni_game}.
\begin{theorem}\label{thm:pun_game}
Given a finite-state deterministic joint policy $\pi=(M, \alpha, \sigma, m_0)$, for any agent $j$, we can construct a simulator for an augmented two-player zero-sum Markov game $\M_j^{\pi}$ (with rewards) which has the following properties.
\begin{itemize}
\item The number of states in $\M_j^{\pi}$ is at most $2|S||M||Q_j|$.
\item The actions of player 1 is $A_j$, and the actions of player 2 is $\A_{-j}=\prod_{i\neq j}A_i$.
\item The min-max value of the two player game corresponds to the deviation cost of $j$, i.e.,
$$\code{dev}_j^{\pi} = \min_{\bar{\pi}_{2}}\max_{\bar{\pi}_1}\bar{J}_j^{\pi}(\bar{\pi}_1,\bar{\pi}_2),$$
where $\bar{J}_j^{\pi}(\bar{\pi}_1,\bar{\pi}_2) = \E\big[\sum_{k=0}^{H} R_j(\bar{s}_k,a_k)\mid \bar{\pi}_1,\bar{\pi}_2\big]$ is the expected sum of rewards w.r.t. the distribution over $(H+1)$-length trajectories generated by using the joint policy $(\bar{\pi}_1,\bar{\pi}_{2})$ in $\M_j^{\pi}$.
\item Given any policy $\bar{\pi}_2$ for player 2 in $\M_j^{\pi}$, we can construct a set of punishment strategies $\tau[j] = \textsc{PunStrat}(\bar{\pi}_2)$ against agent $j$ in $\M$ such that
$$\max_{\bar{\pi}_1}\bar{J}_j^{\pi}(\bar{\pi}_1, \bar{\pi}_2) = \max_{\pi_j'}J_j((\pi\Join\tau[j])_{-j},\pi_j').$$
\end{itemize}
Given an estimate $\tilde{\M}$ of $\M$, we can also construct an estimate $\tilde{\M}_j^{\pi}$ of $\M_j^{\pi}$. 
\end{theorem}
We omit the superscript $\pi$ from $\M_j^{\pi}$ when there is no ambiguity. We denote by $\textsc{ConstructGame}(\tilde{\M},j,\Rc_j,\pi)$ the product construction procedure that constructs and returns $\tilde{\M}_j^{\pi}$.

\begin{algorithm}[t]
\begin{algorithmic}[1]
\STATE $\code{existsNE} \gets \code{True}$
\STATE $\tau \gets \emptyset$
\STATE $\tilde{\M}\gets\textsc{BFS-Estimate}(\M, \delta, p)$ {\color{blue}// Only run if $\M$ has not been estimated before.} 
\FOR{agent $j\in \{1,\ldots,n\}$} 
\STATE $\Rc_j \gets \textsc{ConstructRM}(\p_j)$
\STATE $\tilde{\M}_j \gets \textsc{ConstructGame}(\tilde{\M}, j, \Rc_j, \pi)$
\STATE $\tilde{\code{dev}}_j\gets\min_{\bar{\pi}_{2}}\max_{\bar{\pi}_1}\bar{J}^{\tilde{\M}_j}(\bar{\pi}_1, \bar{\pi}_{2})$
\STATE $\bar{\pi}_2^* \gets \arg\min_{\bar{\pi}_{2}}\max_{\bar{\pi}_1}\bar{J}^{\tilde{\M}_j}(\bar{\pi}_1, \bar{\pi}_{2})$
\STATE $\code{existsNE}\gets \code{existsNE}\wedge(\tilde{\code{dev}}_j \leq J_j(\pi) + \epsilon - \delta)$\label{algline:devcheck}
\STATE $\tau \gets \tau\cup\textsc{PunStrat}(\bar{\pi}_2^*)$
\ENDFOR
\STATE \textbf{return} \code{existsNE}, $\tau$
\caption{\textsc{VerifyNash}\\
Inputs: Finite-state deterministic joint policy $\pi$, specifications $\p_j$ for all $j$, Nash factor $\epsilon$, precision $\delta$, failure probability $p$.\\
Outputs: \code{True} or \code{False} along with a set of punishment strategies $\tau$.}
\label{alg:NEverification}
\end{algorithmic}
\end{algorithm}

\subsection{Solving Min-Max Games}

The min-max game $\M_j$ can be solved using self-play RL algorithms. Many of these algorithms provide probabilistic approximation guarantees for computing the min-max value of the game. We use a model-based algorithm, similar to the one proposed in \cite{pmlr-v119-bai20a}, that first estimates the model $\M_j$ and then solves the game in the estimated model.

One approach is to use existing algorithms for reward-free exploration to estimate the model~\cite{jin2020reward}, but this approach requires estimating each $\M_j$ separately. Under Assumption~\ref{assump:model}, we provide a simpler and more sample-efficient algorithm, called \textsc{BFS-Estimate}, for estimating $\M$. \textsc{BFS-Estimate} performs a search over the transition graph of $\M$ by exploring previously seen states in a breadth first manner. When exploring a state $s$, multiple samples are collected by taking all possible actions in $s$ several times and the corresponding transition probabilities are estimated. After obtaining an estimate of $\M$, we can directly construct an estimate of $\M_j^{\pi}$ for any $\pi$ and $j$ when required. Letting $|Q| = \max_{j}|Q_j|$ and $|M|$ denote the size of the largest finite-state policy output by our enumeration algorithm, we get the following guarantee; see Appendix~\ref{app:bfsestimate} for a proof.
\begin{theorem}\label{thm:estimate}
For any $\delta>0$ and $p\in(0,1]$, \textsc{BFS-Estimate}$(\M,\delta,p)$ computes an estimate $\tilde{\M}$ of $\M$ using $O\left(\frac{|\S|^3|M|^2|Q|^4|\A|H^4}{\delta^2}\log\left(\frac{|\S||\A|}{p}\right)\right)$ sample steps such that with probability at least $1-p$, for any finite-state deterministic joint policy $\pi$ and any agent $j$,
$$\Big|\min_{\bar{\pi}_{2}}\max_{\bar{\pi}_1}\bar{J}^{\tilde{\M}_j^{\pi}}(\bar{\pi}_1, \bar{\pi}_{2}) - \code{dev}_j^{\pi}\Big| \leq \delta,$$
where $\bar{J}^{\tilde{\M}_j^{\pi}}(\bar{\pi}_1,\bar{\pi}_2)$ is the expected reward over length $H+1$ trajectories generated by $(\bar{\pi}_1,\bar{\pi}_2)$ in $\tilde{\M}_j^{\pi}$. Furthermore, letting $\bar{\pi}_2^* \in \arg\min_{\bar{\pi}_{2}}\max_{\bar{\pi}_1}\bar{J}^{\tilde{\M}_j}(\bar{\pi}_1, \bar{\pi}_{2})$ and $\tau[j] = \textsc{PunStrat}(\bar{\pi}_2^*)$, we have
\begin{equation}\label{eq:puneffective}
    \Big|\max_{\bar{\pi}_1}\bar{J}^{\tilde{\M}_j^{\pi}}(\bar{\pi}_1, \bar{\pi}_{2}^*) - \max_{\pi_j'}J_j((\pi\Join\tau[j])_{-j},\pi_j')\Big| \leq \delta.
\end{equation}
\end{theorem}
The min-max value of $\tilde{\M}_j^{\pi}$ as well as $\bar{\pi}_2^*$ can be computed using value iteration. Our full verification algorithm is summarized in Algorithm~\ref{alg:NEverification}. It checks if $\tilde{\code{dev}}_j \leq J_j(\pi) + \epsilon-\delta$ for all $j$, and returns \code{True} if so and \code{False} otherwise. It also simultaneously computes the punishment strategies $\tau$ using the optimal policies for player 2 in the punishment games. Note that \textsc{BFS-Estimate} is called only once (i.e., the first time \textsc{VerifyNash} is called) and the obtained estimate $\tilde{\M}$ is stored and used for verification of every candidate policy $\pi$.
The following soundness guarantee follows from Theorem~\ref{thm:estimate}; proof in Appendix~\ref{app:soundness}.

\begin{corollary}[Soundness]\label{cor:soundness}
For any $p\in(0,1]$, $\varepsilon>0$ and $\delta\in(0,\varepsilon)$, with probability at least $1-p$, if \textsc{HighNashSearch} returns a joint policy $\tilde{\pi}$ then $\tilde{\pi}$ is an $\epsilon$-Nash equilibrium.
\end{corollary}

\section{Complexity}
In this section, we analyze the time and sample complexity of our algorithm in terms of the number of agents $n$, size of the specification $|\p| = \max_{i\in[n]}|\p_i|$, number of states in the environment $|\S|$, number of joint actions $|\A|$, time horizon $H$, precision $\delta$ and the failure probability $p$.

\paragraph{Sample Complexity.} It is known \cite{jothimurugan2021compositional} that the number of edges in the abstract graph $\G_i$ corresponding to specification $\p_i$ is $O(|\p_i|^2)$. Hence for any set of active agents $B$, the number of edges in the product abstract graph $\G_B$ is $O(|\p|^{2|B|})$. Hence total number of edge policies learned by our compositional RL algorithm is $\sum_{B\subseteq [n]}O((|\p|^2)^{|B|}) = O((|\p|^{2}+1)^n)$. We learn each edge using a fixed number of sample steps $C$, which is a hyperparameter.

The number of samples used in the verification phase is the same as the number used by \textsc{BFS-Estimate}. The maximum size of a candidate policy output by the enumeration algorithm $|M|$ is at most the length of the longest path in a product abstract graph. Since the maximum path length in a single abstract graph $\G_i$ is bounded by $|\p_i|$ and at least one agent must progress along every edge in a product graph, the maximum length of a path in any product graph is at most $n|\p|$. Also, the number of states in the reward machine $\Rc_j$ corresponding to $|\p_j|$ is $O(2^{|\p_j|})$. Hence, from Theorem~\ref{thm:estimate} we get that the total number of sample steps used by our algorithm is $O\big((|\p|^{2}+1)^nC + \frac{2^{4|\p|}|\S|^3n^2|\p|^2|\A|H^4}{\delta}\log\big(\frac{|\S||\A|}{p}\big)\big)$.

\paragraph{Time Complexity.} As with sample complexity, the time required to learn all edge policies is $O((|\p|^{2}+1)^n(C+|\A|))$ where the term $|\A|$ is added to account for the time taken to select an action from $\A$ during exploration (we use $Q$-learning with $\varepsilon$-greedy exploration for learning edge policies). Similarly, time taken for constructing the reward machines and running \textsc{BFS-Estimate} is $O(\frac{2^{4|\p|}|\S|^3n^2|\p|^2|\A|H^4}{\delta}\log\big(\frac{|\S||\A|}{p}\big))$.

The total number of path policies considered for a given set of active agents $B$ is bounded by the number of paths in the product abstract graph $\G_B$ that terminate in a final product state. First, let us consider paths in which exactly one agent progresses in each edge. The number of such paths is bounded by $(|B||\p|)^{|B||\p|}$ since the length of such paths is bounded by $|B||\p|$ and there are at most $|B||\p|$ choices at each step---i.e., progressing agent $j$ and next vertex of the abstract graph $\G_{\p_j}$. Now, any path in $\G_B$ can be constructed by merging adjacent edges along such a path (in which at most one agent progresses at any step). The number of ways to merge edges along such a path is bounded by the number of groupings of edges along the path into at most $|B||\p|$ groups which is bounded by $(|B||\p|)^{|B||\p|}$. Therefore, the total number of paths in $\G_B$ is at most $2^{2|B||\p|\log(n|\p|)}$. Finally, the total number of path policies considered is at most $\sum_{B\subseteq[n]}2^{2|B||\p|\log(n|\p|)} \leq ((n|\p|)^{2|\p|}+1)^n = O(2^{2n|\p|\log(2n|\p|)})$.

Now, for each path policy $\pi$, the verification algorithm solves $\tilde{\M_j^{\pi}}$ using value iteration which takes $O(|\tilde{\S}||\A|Hf(|\A|)) = O(2^{|\p|}n|\p||\S||\A|Hf(|\A|))$ time, where $f(|\A|)$ is the time required to solve a linear program of size $|\A|$. Also accounting for the time taken to sort the path policies, we arrive at a time complexity bound of $2^{O(n|\p|\log(n|\p|))}\text{poly}(|\S|,|\A|,H,\frac{1}{p}, \frac{1}{\delta})$.

It is worth noting that the procedure halts as soon as our verification procedure successfully verifies a policy; this leads to early termination for cases where there is a high value $\epsilon$-Nash equilibrium (among the policies considered). Furthermore, our verification algorithm runs in polynomial time and therefore one could potentially improve the overall time complexity by reducing the search space in the prioritized enumeration phase---e.g., by using domain specific insights.

\section{Experiments}
\label{Sec:experiments}
We evaluate our algorithm on finite state environments and a variety of specifications, aiming to answer the following:
\begin{itemize}
\item Can our approach be used to learn $\epsilon$-Nash equilibria?
\item Can our approach learn policies with high social welfare?
\end{itemize}
We compare our approach to two baselines described below, using two metrics: (i) the social welfare $\code{welfare}(\pi)$ of the learned joint policy $\pi$, and (ii) an estimate of the minimum value of $\epsilon$ for which $\pi$ forms an $\epsilon$-Nash equilibrium:
$$\epsilon_{\min}(\pi) = \max\{J_i(\pi_{-i}, \bestr_i(\pi)) - J_i(\pi)\mid i\in[n]\}.$$
Here, $\epsilon_{\min}(\pi)$ is computed using single agent RL (specifically, $Q$-learning) to compute $\bestr_i(\pi)$ for each agent $i$.

\paragraph{Environments and specifications.}

We show results on the \emph{Intersection environment} illustrated in Figure~\ref{fig:intersection}, which consists of $k$-cars (agents) at a 2-way intersection of which $k_1$ and $k_2$ cars are placed along the N-S and E-W axes, respectively. The state consists of the location of all cars where the location of a single car is a non-negative integer. 1 corresponds to the intersection, 0 corresponds to the location one step towards the south or west of the intersection (depending on the car) and locations greater than 1 are to the east or north of the intersection. Each agent has two actions. \code{STAY} stays at the current position. \code{MOVE} decreases the position value by 1 with probability 0.95 and stays with probability 0.05. We consider specifications similar to the ones in the motivating example. Details are in Appendix~\ref{sec:benchmarks}, and results on two additional environments are in Appendix~\ref{sec:add_results}.

\paragraph{Baselines.}

We compare our NE computation method (\nash) to two approaches for learning in non-cooperative games. The first, \maqrm, is an adaption of the reward machine based learning algorithm proposed in \cite{neary2021reward}.
\maqrm was originally proposed for cooperative multi-agent RL where there is a single specification for all the agents. It proceeds by first decomposing the specification into individual ones for all the agents and then runs a Q-learning-style algorithm (\textsc{qrm}) in parallel for all the agents. We use the second part of their algorithm directly since we are given a separate specification for each agent.
The second baseline, \nvi, is a model-based approach that first estimates transition probabilities, and then computes a Nash equilibrium in the estimated game using value iteration for stochastic games~\cite{kearns2000fast}. To promote high social welfare, we select the highest value Nash solution for the matrix game at each stage of value iteration. Note that this greedy strategy may not maximize social welfare.  
Both \maqrm and \nvi learn from rewards as opposed to specification; thus, we supply rewards in the form of reward machines constructed from the specifications. $\nvi$ is guaranteed to return an $\epsilon$-Nash equilibrium with high probability, but $\maqrm$ is not guaranteed to do so. Details are in Appendix~\ref{app:baselines}.

\begin{table}[t]
\centering\scriptsize
\begin{tabular}{cccrrrr}
\toprule
\multirow{3}{*}{Spec.} &
\multirow{3}{*}{\begin{tabular}[c]{@{}c@{}}Num. of\\ agents\end{tabular}} &
\multirow{3}{*}{Algorithm} &
\multirow{3}{*}{\begin{tabular}[c]{@{}c@{}}\qquad$\code{welfare}(\pi)$\qquad\\ \qquad(avg \textpm ~ std)\qquad\end{tabular}} &
\multirow{3}{*}{\begin{tabular}[c]{@{}c@{}}\qquad$\epsilon_{\min}(\pi)$\qquad\\ \qquad(avg \textpm ~ std)\qquad\end{tabular}} &
\multirow{3}{*}{\begin{tabular}[c]{@{}c@{}}Num. of\\ terminated \\runs\end{tabular}} &
\multirow{3}{*}{\begin{tabular}[c]{@{}c@{}}Avg. num. of\\ sample steps\\(in millions)\end{tabular}} \\
{} & {} & {} & {} & {} & {}\\
{} & {} & {} & {} & {} & {}\\
\midrule
\multirow{3}{*}{$\p^1$} &
\multirow{3}{*}{3} &
\nash & \textbf{0.33 \textpm ~ 0.00} & \textbf{0.00 \textpm ~ 0.00} & 10 & 1.78 \\
{} & {} & \textsc{nvi} & 0.32 \textpm ~ 0.00 & \textbf{0.00 \textpm ~ 0.00} & 10 & 1.92 \\
{} & {} & \textsc{maqrm} & 0.18 \textpm ~ 0.01 & 0.51 \textpm ~ 0.01 & 10 & 2.00 \\

\midrule
\multirow{3}{*}{$\p^2$} &
\multirow{3}{*}{4} &
\nash & \textbf{0.55 \textpm ~ 0.10} & \textbf{0.01 \textpm ~ 0.02} & 10 & 11.53 \\
{} & {} & \textsc{nvi} & 0.04 \textpm ~ 0.01 & 0.02 \textpm ~ 0.01 & 10 & 12.60 \\
{} & {} & \textsc{maqrm} &  0.12 \textpm ~ 0.01 & 0.20 \textpm ~ 0.03 & 10 & 15.00 \\

\midrule
\multirow{3}{*}{$\p^3$} &
\multirow{3}{*}{4} &
\nash & \textbf{0.49 \textpm ~ 0.01} & \textbf{0.00 \textpm ~ 0.01} & 10 & 11.26 \\
{} & {} & \textsc{nvi} & 0.45 \textpm ~ 0.01 & \textbf{0.00 \textpm ~ 0.01} & 10 & 12.60 \\
{} & {} & \textsc{maqrm} & 0.11 \textpm ~ 0.01 & 0.22 \textpm ~ 0.02 & 10 & 15.00 \\

\midrule
\multirow{3}{*}{$\p^4$} &
\multirow{3}{*}{3} &
\nash & 0.90 \textpm ~ 0.15 & \textbf{0.00 \textpm ~ 0.00} & 10 & 2.16 \\
{} & {} & \textsc{nvi} & \textbf{0.98 \textpm ~ 0.00} & \textbf{0.00 \textpm ~ 0.00} & 4 & 2.18 \\
{} & {} & \textsc{maqrm} & 0.23 \textpm ~ 0.01 & 0.39 \textpm ~ 0.04 & 10 & 2.00 \\

\midrule
\multirow{3}{*}{$\p^5$} &
\multirow{3}{*}{5} &
\nash & \textbf{0.58 \textpm ~ 0.02} & \textbf{0.00 \textpm ~ 0.00} & 10 & 62.17 \\
{} & {} & \textsc{nvi} & 0.05 \textpm ~ 0.01 & 0.01 \textpm ~ 0.01 & 7 & 80.64 \\
{} & {} & \textsc{maqrm} & Timeout & Timeout & 0 & Timeout  \\
\bottomrule\\

\end{tabular}
\caption{Results for all specifications in Intersection Environment. Total of 10 runs per benchmark. Timeout = 24 hrs.}
\label{tab:int_results}
\end{table}

\paragraph{Results.}

Our results are summarized in Table~\ref{tab:int_results}. For each specification, we ran all algorithms 10 times with a timeout of 24 hours. Along with the average social welfare and $\epsilon_{\min}$, we also report the average number of sample steps taken in the environment as well as the number of runs that terminated before timeout. For a fair comparison, all approaches were given a similar number of samples from the environment.

\paragraph{Nash equilibrium.} Our approach learns policies that have low values of $\epsilon_{\min}$, indicating that it can be used to learn $\epsilon$-Nash equilibria for small values of $\epsilon$. \nvi also has similar values of $\epsilon$, which is expected since \nvi provides guarantees similar to our approach w.r.t. Nash equilibria computation. On the other hand, \maqrm learns policies with large values of $\epsilon_{\min}$, implying that it fails to converge to a Nash equilibrium in most cases.

\paragraph{Social Welfare.} Our experiments show that our approach consistently learns policies with high social welfare compared to the baselines. For instance, $\p^3$ corresponds to the specifications in the motivating example for which our approach learns a joint policy that causes both blue and black cars to achieve their goals. Although \nvi succeeds in learning policies with high social welfare for some specifications ($\p^1,\p^3$, $\p^4$), it fails to do so for others ($\p^2$, $\p^5$). {Additional experiments (see Appendix~\ref{sec:add_results}) indicate that \nvi achieves similar social welfare as our approach for specifications in which all agents can successfully achieve their goals (cooperative scenarios). However, in many other scenarios in which only some of the agents can fulfill their objectives, our approach achieves higher social welfare.}


\section{Conclusions}

We have proposed a framework for maximizing social welfare under the constraint that the joint policy should form an $\epsilon$-Nash equilibrium. Our approach involves learning and enumerating a small set of finite-state deterministic policies in decreasing order of social welfare and then using a self-play RL algorithm to check if they can be extended with punishment strategies to form an $\epsilon$-Nash equilibrium. Our experiments demonstrate that our approach is effective in learning Nash equilibria with high social welfare.

One limitation of our approach is that {our algorithm does not have any guarantee regarding optimality with respect to social welfare}. The policies considered by our algorithm are chosen heuristically based on the specifications, which may lead to scenarios where we miss high welfare solutions. For example, $\p^2$ corresponds to specifications in the motivating example except that the blue car is not required to stay a car length ahead of the other two cars. In this scenario, it is possible for three cars to achieve their goals in an equilibrium solution if the blue car helps the cars behind by staying in the middle of the intersection until they catch up. Such a joint policy is not among the set of policies considered; therefore, our approach learns a solution in which only two cars achieve their goals. We believe that such limitations can be overcome in future work by modifying the various components within our enumerate-and-verify framework.

\paragraph{\rm \textbf{Acknowledgements.}}
We thank the anonymous reviewers for their helpful comments. This work is supported in part by NSF grant 2030859 to the CRA for the CIFellows Project, ONR award N00014-20-1-2115,
DARPA Assured Autonomy award, NSF award CCF 1723567 and ARO award W911NF-20-1-0080.

%
%
\bibliographystyle{splncs04}
\bibliography{main}

\begin{thebibliography}{10}
\providecommand{\url}[1]{\texttt{#1}}
\providecommand{\urlprefix}{URL }
\providecommand{\doi}[1]{https://doi.org/#1}

\bibitem{akchurina2008multi}
Akchurina, N.: Multi-agent reinforcement learning algorithm with variable
  optimistic-pessimistic criterion. In: ECAI. vol.~178, pp. 433--437 (2008)

\bibitem{alur2021framework}
Alur, R., Bansal, S., Bastani, O., Jothimurugan, K.: A framework for
  transforming specifications in reinforcement learning. arXiv preprint
  arXiv:2111.00272  (2021)

\bibitem{pmlr-v119-bai20a}
Bai, Y., Jin, C.: Provable self-play algorithms for competitive reinforcement
  learning. In: Proceedings of the 37th International Conference on Machine
  Learning (2020)

\bibitem{bouyer2010nash}
Bouyer, P., Brenguier, R., Markey, N.: Nash equilibria for reachability
  objectives in multi-player timed games. In: International Conference on
  Concurrency Theory. pp. 192--206. Springer (2010)

\bibitem{rmax}
Brafman, R.I., Tennenholtz, M.: R-max - a general polynomial time algorithm for
  near-optimal reinforcement learning. JMLR  \textbf{3} (2003)

\bibitem{chatterjee2005two}
Chatterjee, K.: Two-player nonzero-sum $\omega$-regular games. In:
  International Conference on Concurrency Theory. pp. 413--427. Springer (2005)

\bibitem{chatterjee2004nash}
Chatterjee, K., Majumdar, R., Jurdzi{\'n}ski, M.: On nash equilibria in
  stochastic games. In: International workshop on computer science logic. pp.
  26--40. Springer (2004)

\bibitem{czumaj2015approximate}
Czumaj, A., Fasoulakis, M., Jurdzinski, M.: Approximate nash equilibria with
  near optimal social welfare. In: Twenty-Fourth International Joint Conference
  on Artificial Intelligence (2015)

\bibitem{greenwald2003correlated}
Greenwald, A., Hall, K., Serrano, R.: Correlated q-learning. In: ICML. vol.~3,
  pp. 242--249 (2003)

\bibitem{hammond2021}
Hammond, L., Abate, A., Gutierrez, J., Wooldridge, M.: Multi-agent
  reinforcement learning with temporal logic specifications. In: International
  Conference on Autonomous Agents and MultiAgent Systems. p. 583–592 (2021)

\bibitem{best-nash}
Hazan, E., Krauthgamer, R.: How hard is it to approximate the best nash
  equilibrium? In: Proceedings of the Twentieth Annual ACM-SIAM Symposium on
  Discrete Algorithms. p. 720–727. SODA '09, Society for Industrial and
  Applied Mathematics (2009)

\bibitem{hopcroft2001introduction}
Hopcroft, J.E., Motwani, R., Ullman, J.D.: Introduction to automata theory,
  languages, and computation. Acm Sigact News  \textbf{32}(1),  60--65 (2001)

\bibitem{hu2003nash}
Hu, J., Wellman, M.P.: Nash q-learning for general-sum stochastic games.
  Journal of machine learning research  \textbf{4}(Nov),  1039--1069 (2003)

\bibitem{hu1998multiagent}
Hu, J., Wellman, M.P., et~al.: Multiagent reinforcement learning: theoretical
  framework and an algorithm. In: ICML. vol.~98, pp. 242--250. Citeseer (1998)

\bibitem{icarte2018using}
Icarte, R.T., Klassen, T., Valenzano, R., McIlraith, S.: Using reward machines
  for high-level task specification and decomposition in reinforcement
  learning. In: International Conference on Machine Learning. pp. 2107--2116.
  PMLR (2018)

\bibitem{jin2020reward}
Jin, C., Krishnamurthy, A., Simchowitz, M., Yu, T.: Reward-free exploration for
  reinforcement learning. In: International Conference on Machine Learning. pp.
  4870--4879. PMLR (2020)

\bibitem{jothimurugan2021compositional}
Jothimurugan, K., Bansal, S., Bastani, O., Alur, R.: Compositional
  reinforcement learning from logical specifications. Advances in Neural
  Information Processing Systems  \textbf{34} (2021)

\bibitem{kearns2000fast}
Kearns, M., Mansour, Y., Singh, S.: Fast planning in stochastic games. In:
  Proceedings of the Sixteenth conference on Uncertainty in artificial
  intelligence. pp. 309--316 (2000)

\bibitem{kwiatkowska2019equilibria}
Kwiatkowska, M., Norman, G., Parker, D., Santos, G.: Equilibria-based
  probabilistic model checking for concurrent stochastic games. In:
  International Symposium on Formal Methods. pp. 298--315. Springer (2019)

\bibitem{kwiatkowska2020prism}
Kwiatkowska, M., Norman, G., Parker, D., Santos, G.: Prism-games 3.0:
  Stochastic game verification with concurrency, equilibria and time. In:
  International Conference on Computer Aided Verification. pp. 475--487.
  Springer (2020)

\bibitem{littman1994markov}
Littman, M.L.: Markov games as a framework for multi-agent reinforcement
  learning. In: Machine learning proceedings 1994, pp. 157--163. Elsevier
  (1994)

\bibitem{littman2001friend}
Littman, M.L.: Friend-or-foe q-learning in general-sum games. In: ICML. vol.~1,
  pp. 322--328 (2001)

\bibitem{gambit}
McKelvey, R.D., McLennan, A.M., Turocy, T.L.: Gambit: Software tools for game
  theory (2014), \url{http://www.gambit-project.org}

\bibitem{neary2021reward}
Neary, C., Xu, Z., Wu, B., Topcu, U.: Reward machines for cooperative
  multi-agent reinforcement learning (2021)

\bibitem{pmlr-v54-perolat17a}
Perolat, J., Strub, F., Piot, B., Pietquin, O.: {Learning Nash Equilibrium for
  General-Sum Markov Games from Batch Data}. In: Proceedings of the 20th
  International Conference on Artificial Intelligence and Statistics (2017)

\bibitem{prasad2015two}
Prasad, H., LA, P., Bhatnagar, S.: Two-timescale algorithms for learning nash
  equilibria in general-sum stochastic games. In: Proceedings of the 2015
  International Conference on Autonomous Agents and Multiagent Systems. pp.
  1371--1379 (2015)

\bibitem{shapley1953stochastic}
Shapley, L.S.: Stochastic games. Proceedings of the national academy of
  sciences  \textbf{39}(10),  1095--1100 (1953)

\bibitem{tor-etal-arxiv20}
Toro~Icarte, R., Klassen, T.Q., Valenzano, R., McIlraith, S.A.: Reward
  machines: Exploiting reward function structure in reinforcement learning.
  arXiv preprint arXiv:2010.03950  (2020)

\bibitem{wei2017online}
Wei, C.Y., Hong, Y.T., Lu, C.J.: Online reinforcement learning in stochastic
  games. In: Proceedings of the 31st International Conference on Neural
  Information Processing Systems. pp. 4994--5004 (2017)

\bibitem{zinkevich2006cyclic}
Zinkevich, M., Greenwald, A., Littman, M.: Cyclic equilibria in markov games.
  Advances in Neural Information Processing Systems  \textbf{18}, ~1641 (2006)

\end{thebibliography}
%
\clearpage 
\clearpage
\onecolumn
\appendix

\sloppy

\section{Prioritized Enumeration}
\label{app:search}


\subsection{Proof of Theorem~\ref{thrm:productgraph}}



\begin{proof}
We show that if $\zeta \models_B \rho$ then every agent $i \in B$ achieves a path in the abstract graph $\G_i$ of its specification $\p_i$. Let $0=k_0 \leq k_1 \leq \cdots \leq k_\ell \leq t$ be the indices that satisfy the criteria in Definition~\ref{def:achievepath}.

For agent $i\in B$, let indices $0 \leq z_0 < \cdots < z_x < \ell$ be such that the agent makes progress along the edge $e_{z_y} = \ol{u}_{z_y} \to \ol{u}_{z_{y} +1}$ for $0 \leq y \leq x$. Note that agents can make progress only till they reach a final state in their abstract graph. So, $(u_{z_y})_i \notin F_i$ and $(u_{z_x +1})_i \in F_i$ for all $0 \leq y \leq x$.  Further note, $(u_{z_0})_i = u_0^i$ and  $ (u_{z_{y+1}})_i = (u_{z_y +1})_i $ for $0\leq y < x$ since the agent $i$ has not made any progress in between.

Let $\zeta_y = \zeta_{k_{z_y}: k_{z_y +1}}$ be the sub-trajectory that achieves the edge $e_{z_y}$ for $0 \leq y \leq x$, i.e. $\zeta_y \models_B e_{z_y}$ for $0 \leq y \leq x$.

Since, agent $i$ has made progress along $e_{z_y}$, using Definition~\ref{def:edgerollout}, we can obtain  indices $0= p_0\leq\dots<  p_{x+1}\leq t$ on the trajectory such that 
\begin{itemize}
    \item $p_y \leq k_{z_y} \leq p_{y+1}$ for $0\leq y < x$
    \item $s_{p_{y+1}} \in \beta_i((u_{z_y+1})_i)$ for all $0 \leq y \leq x$

    \item $\zeta_{k_{z_y}: p_{y+1}} \in \Zc_{\safe, i}^{(e_{z_y})_i}$ for all $0 \leq y \leq x$
    \item $\zeta_{p_{y+1}: k_{{z_y}+1}} \in \first(\Zc_{\safe, i}^{(u_{z_y +1})_i \to w_i})$ for $w_i \in \outgoing((u_{z_y +1})_i)$ for $0\leq y < x$ and $\zeta_{p_{x+1}:k_{z_x +1}} \models \Zc_{\safe, i}^{(u_{z_x +1})_i}$
    
    (This is because $(u_{z_y})_i \notin F_i$ and $(u_{z_x +1})_i \in F_i$ for all $0 \leq y \leq x$, as observed earlier)
\end{itemize}

Additionally, by using, $(u_{z_y +1})_i = (u_{z_{y+1}})_i$ for $0\leq y < x$ since the agent $i$ has not made any progress in between, the above is simplified to: 
there exists  indices $0= p_0\leq\dots< p_{x+1} \leq t$ on the trajectory such that

\begin{itemize}
    \item $p_y \leq k_{z_y} \leq p_{y+1}$ for $0\leq y \leq x$
    \item $s_{p_{y+1}} \in \beta_i((u_{z_{y+1}})_i)$ for all $0 \leq y \leq x$

    \item $\zeta_{k_{z_y}: p_{y+1}} \in \Zc_{\safe, i}^{(e_{z_y})_i}$ for all $0 \leq y \leq x$
    \item $\zeta_{p_{y+1}: k_{{z_y}+1}} \in \first(\Zc_{\safe, i}^{(e_{z_{y+1}})_i})$  for $0\leq y < x$ and $\zeta_{p_{x+1}:k_{z_x +1}} \models \Zc_{\safe, i}^{(u_{z_x +1})_i}$
    
    (This is because $(e_{z_{y+1}})_i$ is an outgoing edge from $(u_{z_{y+1}})_i$)
\end{itemize}

\paragraph{}

Our goal is to show that $\zeta_{p_y:p_{y+1}} \in \Zc_{\safe, i}^{(e_{z_y})_i}$ for $0\leq y \leq x$.
We already know that $\zeta_{k_{z_y}: p_{y+1}} \in \Zc_{\safe, i}^{(e_{z_y})_i}$ for all $0 \leq y \leq x$.
Observing the fact that for any edge $e$ of $\G_i$ the set $\first(\Zc_{\safe,i}^e)$ is closed under concatenation, it is sufficient to show $\zeta_{p_y:k_{z_y}} \in \first(\Zc_{\safe, i}^{(e_{z_y})_i})$ for $0\leq y \leq x$. We do this in two cases. \begin{itemize}
    \item \textbf{Case $y = 0$}: In this case $\zeta_{p_0:k_{z_0}} = \zeta_{k_0:k_{z_0}}$. So, it has made no progress along the path till $s_{k_{z_0}}$. So, the non-progressing agent will ensure its trajectory is safe with respect to outgoing edges from vertex $(u_{z_0})_i = (u_0)_i$. Now, $(e_{z_0})_i$ is an outgoing edge from $(u_{z_0})_i$. So, we get that $\zeta_{0:k_{z_0}} \in \first(\Zc_{\safe, i}^{(e_{z_0})_i})$.
    
    \item \textbf{Case $0<y\leq x$}: Here, we know that $\zeta_{p_{y}: k_{{z_{y-1}+1} }} \in \first(\Zc_{\safe, i}^{(e_{z_y})_i})$.  So, it is sufficient to show that $\zeta_{k_{{z_{y-1}}+1}: k_{z_y}} \in  \first(\Zc_{\safe, i}^{(e_{z_y})_i}).$ This is true because agent $i$ is has not progressed on any of the   edges between $\ol{u}_{{{z_{y-1}}+1}}$ and $\ol{u}_{{z_y}}$. The $i^{\text{th}}$ component of all vertices between these states is $(u_{{z_y}})_i$, since agent $i$ will not change its vertex. Now $(e_{z_y})_i$ is an outgoing edge from $(u_{z_y})_i$. So, in particular, we get that $\zeta_{k_{{z_{y-1}}+1}: k_{z_y}} \in  \first(\Zc_{\safe, i}^{(e_{z_y})_i})$.

\end{itemize}

\paragraph{}

So, consider the path $(u_0)_i\to (u_{{z_0}})_i \to  (u_{{z_1}})_i \to (u_{{z_x}})_i \to (u_{{z_x} + 1})_i \in F_i$ in graph $\G_i$. 
There exists indices $0= p_0\leq p_1<...<p_{x+1}\leq t$ such that 
\begin{itemize}
    \item $s_{p_0} = s_0 \in \beta_i((u_0)_i)$ (from Definition~\ref{def:achievepath}), $s_{p_y} \in \beta_i((u_{{z_y}})_i)$ for $0<y\leq x$ and $s_{p_{x+1}} \in \beta_i((u_{{z_{x} + 1}})_i)$ (from inferences made above)
    
    \item $\zeta_{p_y:p_{y+1}} \in \Zc_{\safe, i}^{(e_{z_y})_i}$ for $0\leq y \leq x$
    
    \item $\zeta_{p_{x+1}:t} \in \Zc_{\safe,i}^{(u_{{z_x} + 1})_i }$ because agent $i$ cannot progress any further after visiting state ${(u_{{z_x} + 1})_i}\in F_i$.
    
\end{itemize}

Thus, for agent $i$,  $\zeta \models \G_i$, which implies $\zeta \models \p_i$.\qed
\end{proof}

\subsection{Algorithm}

The complete algorithm for prioritized enumeration is outlined in Algorithm~\ref{Alg:PriortizedEnum}.
The details of non-standard functions are given below. 

\begin{algorithm}[t]
	\caption{\textsc{PrioritizedEnumeration}\\
	\textbf{Inputs:} $n$-agent Environment $\M$ and agent specifications $\p_1, \dots, \p_n$ \\
	\textbf{Output:} Ranking scheme}
	\label{Alg:PriortizedEnum}
	
	\begin{algorithmic}[1]
	
	\STATE{Initialize $\pathcostmaxheap \gets \varnothing$}    
	    \STATE{\textbf{for} $i \in [n]$ \textbf{do} $\G_{i} \gets \AbstractGraph(\p_i)$}\label{algline:learn:ConstructAbstractGraph}

        \FOR{$B \in 2^{[n]}$}	
            \STATE{$\G_B \gets \MakeProduct(\G_1, \dots \G_n, B)$}
            \STATE{Initialize path policies $\Pi(\ol{u}_0) \gets \{\varepsilon\}$ and $\Pi(\ol{u}) \gets \varnothing$ if $\ol{u}\neq\ol{u}_0$}
            
            \STATE{Initialize state distribution $\Gamma(\ol{u}_0) \gets \{\mathsf{At}(s_0)\}$ and $\Gamma(\ol{u}) \gets \varnothing$ if $\ol{u}\neq\ol{u}_0$}
           
           
            \STATE{$\TSortList\gets \TopoSort(\ol{U}, \ol{E})$}
            \WHILE{$\TSortList \neq \varnothing$}
                \STATE{$\ol{u} \gets \TSortList.\pop()$}
                \STATE{$\eta_{\ol{u}} \gets \StateDistribution(\Gamma(\ol{u}))$} 
                \FOR{$e = \ol{u} \to \ol{v}\in \outgoing(\ol{u})$}
                    \STATE {$\pi_e \gets \learnpolicy(e, \eta_{\ol{u}})$}
                    \STATE{$\eta_{\ol{v},e} \gets \reachdistribution(e, \pi_e, \eta_{\ol{u}})$}
                    \STATE{Add $\eta_{\ol{v},e}$ to $\Gamma(\ol{v})$}
                    \STATE{\textbf{for} $\pi_\rho \in \Pi(\ol{u})$ \textbf{do} Add $\pi_\rho\circ\pi_e$ to $\Pi(\ol{v})$}
                \ENDFOR
                \IF{$\ol{u} \in F$}
                    \FOR{$\pi_\rho\in \Pi(\ol{u})$}
                        \STATE{$c \gets \estimatecost(\pi_\rho, \p_1,\dots, \p_n)$}
                        \STATE{Add $(\pi_\rho, c)$ to $\pathcostmaxheap$}
                    \ENDFOR 
                \ENDIF 
            \ENDWHILE
        \ENDFOR
        \RETURN{$\pathcostmaxheap$}
	\end{algorithmic}
\end{algorithm}


\paragraph{$\StateDistribution$:} The set $\Gamma(\ol{u})$ consists of as many state distributions as there are incoming edges into the state when $\ol{u} \neq \ol{u_0}$. When $\ol{u} = \ol{u}_0$, $\Gamma(\ol{u})$ only contains the distribution corresponding to the initial state distribution of the underlying environment. The function $\StateDistribution$ computes an initial state distribution for the input abstract product vertex $\ol{u}$ by taking an average of all of the distributions in $\Gamma(\ol{u})$.

\paragraph{$\learnpolicy$:} Use (single agent) RL to learn a co-operative joint policy $\pi_e$ so that $\pi_e$ achieves the edge $e=\ol{u}\to \ol{v}$ from the given initial state distribution $\eta_{\ol{u}}$.  We use single-agent RL, specifically Q-learning, to learn a co-operative joint policy to achieve an edge with the reward $\mathbbm{1}(\zeta\models_B e)$. 
Precisely, we learn $\pi_e$ such that
$$  \pi_e \in \operatorname*{\arg\max}_{\pi} \Pr_{s_0\sim\eta_{\ol{u}}, \zeta\sim\mathcal{D}_{\pi, s_0}}\left[\zeta\models_B e\right]$$ 
where $\zeta\sim\mathcal{D}_{\pi_e, s_0}$ is the trajectory sampled from executing policy $\pi_e$ from state $s_0$.

\paragraph{$\reachdistribution$:} Given an edge $e = \ol{u}\to\ol{v}$, edge policy $\pi_e$, and initial state distribution  $\eta_{\ol{u}}$, this function evaluates the state distribution induced on $\ol{v}$ upon executing policy $\pi_e$ with an initial state distribution $\eta_{\ol{u}}$. Formally, for any $s\in\S$
$$\Pr_{s'\sim\eta_{\ol{v},e}}[s=s'] = \Pr_{s_0\sim\eta_{\ol{u}}, \zeta\sim\mathcal{D}_{\pi_e, s_0}}\left[s=s_{k}\mid \zeta_{0:k}\models_B e \ \text{and}\ \forall k'<k, \zeta_{0:k'}\not\models_B e\right]$$ 
where $\zeta\sim\mathcal{D}_{\pi_e, s_0}$ is the trajectory sampled from executing policy $\pi_e$ from state $s_0$ and $k$ is the length of the smallest prefix of $\zeta$ that achieves $e$.

\paragraph{$\estimatecost$:} Once a path policy $\pi_{\rho}$ is learnt, we estimate the probability of satisfaction of the specifications $J_i(\pi_\rho)$ for all agents $i$ using Monte-Carlo sampling. Then $\code{welfare}(\pi_\rho)$ is computed by taking the mean of the probabilities of satisfaction of agent specifications.

\section{Nash Equilibria Verification}\label{app:verify}

\subsection{Best Punishment Strategy}\label{sec:best_puni}
Our verification procedure involves computing the best punishment strategies $\tau^*[j]$ against agent $j$ for each $j\in[n]$. We give a proof of Theorem~\ref{thm:best_pun} below from which it follows that computing the best punishment strategies is sufficient to decide whether a given finite-state deterministic joint policy $\pi$ can be extended to an $\epsilon$-Nash equilibrium.

\begin{proof}[Proof of Theorem~\ref{thm:best_pun}]
It follows from the definition of $\pi\Join\tau$ that the distribution over $H$-length trajectories induced by $\pi\Join\tau$ in $\M$ is the same as the one induced by $\pi$ since $\tau$ is never triggered when all agents are following $\pi$. Therefore, $J_j(\pi\Join\tau) = J_j(\pi)$ for all $j$ and $\tau$. Now, suppose there is a $\tau$ such that $\pi\Join\tau$ is an $\epsilon$-Nash equilibrium. Then for all $j$,
$$J_j((\pi\Join\tau)_{-j}, \bestr_j(\pi\Join\tau))\leq J_j(\pi\Join\tau)+\epsilon = J_j(\pi) + \epsilon.$$
But $\tau^*[j]$ minimizes the LHS of the above equation which is independent of $\tau\setminus\tau[j]$. Therefore, for all $j$,
\begin{align*}
    J_j((\pi\Join\tau^*)_{-j},\bestr_j(\pi\Join\tau^*))&\leq J_j((\pi\Join\tau)_{-j},\bestr_j(\pi\Join\tau))\\
    &\leq J_j(\pi) + \epsilon\\
    &= J_j(\pi\Join\tau^*)+\epsilon.
\end{align*}
Hence, $\pi\Join\tau^*$ is an $\epsilon$-Nash equilibrium. The rest of the Theorem follows from the definition of $\epsilon$-Nash equilibrium.\qed
\end{proof}
\newcommand{\initial}{\mathsf{init}}
\newcommand{\form}{\mathsf{Formula}}
\newcommand{\dead}{\mathsf{dead}}

\subsection{Reward Machine Construction}
\label{appendix:rm}

In this section, we detail the construction of reward machines from $\spectrl$ specifications such that the reward of any finite-length trajectory is 1 if the trajectory satisfies the specification and 0 otherwise. 

We proceed by constructing deterministic finite-state automata (DFA) that accepts all trajectories which satisfy the specification. Next, we will convert the DFA into a reward machine with the desired reward function. 

\paragraph{DFA Construction.}


A finite-state automaton is a tuple $D = (Q, \B, \delta, q_{\initial}, F)$ where $Q$ is  a finite-set of states, $\B$ is a finite-set of propositions, $q_{\initial}$ is the initial state, and $F \subseteq Q$ is the set of accepting states. The transition relation is defined as $\delta \subseteq Q \times \form(\B) \times Q$ where $\form(\B)$ is the set of boolean formulas over propositions $\B$.
A finite-state automaton is {\em deterministic} if every assignment $\sigma \in 2^\P$ can transition to a unique state from every state, i..e, if for all states $q\in Q$ and assignments $\sigma \in 2^\B$,   $|\{q' \mid (q, b, q') \in \delta \text{ and } \sigma \models b  \}| \leq 1$. Otherwise, it is {\em non-deterministic}. 
Every non-deterministic finite-state automata (NFA) can be converted to a deterministic finite-state automata (DFA).
A {\em run} of a {\em word} (sequence of assignments over $\B$) given by $w = w_0\dots w_m \in (2^{\B})^*$ is a sequence of states $\rho = q_0\dots q_{m+1}$ such that $q_0 = q_\initial$ and there exists $(q_i, b_i, q_{i+1}) \in \delta$ such that $w_{i}\models b_i$ for all $0\leq i <m$.
A run $q_0\dots q_{m+1}$ is accepting if $q_{m+1} \in F$. A word $w$ is accepted by $D$ if it has an accepting run.


Let $\spectrl$ specifications be defined over the set of basic predicates $\P_0$.We define a labelling function $L : \S \to 2^{\P_0}$ such that $L(s) = \{p \mid \semantics{p}(s) =\true\}$. Given a trajectory $\zeta = s_0,\dots, s_t$ in the environment $\M$, let its proposition sequence $\L(\zeta)$ be given by $\L(s_0),\dots, \L(s_t)$.

\begin{lemma}
Given $\spectrl$ specification $\p$, we can construct a DFA $D_\p$ such that a trajectory $\zeta \models \p$ iff $\L(\zeta)$ is accepted by $D_\p$.
\end{lemma}
\begin{proof}
We use structural induction on \spectrl specifications to construct the desired DFA.
The construction is very similar to the construction of finite-state automata from regular expressions, and hence details of proof have been skipped~\cite{hopcroft2001introduction}. The construction is given below:

\paragraph{Eventually (\rm $\p ::= \eventually{b} $).}

Construct finite-state automata $D_\p = (\{q_\initial, q\}, \P_0, \delta, q_\initial, \{q\})$ where $$\delta = \{(q_\initial, \neg b, q_\initial), (q_\initial, b, q), (q, \mathsf{True}, q) \}.$$

Clearly, $D_\p$ is deterministic because the only state from which more that two transitions emanate is $q_\initial$ is defined on functions that negate one another ($b$ and $\neg b$), hence they will not have any common assignment. 



\paragraph{Always (\rm $\p ::=  \p_1 \always{b}$.)}

Let the DFA for $\p_1$ be $D_1 = (Q, \P_0, \delta_1, q_\initial, F)$. 
Then the DFA for $\p$ is given by  $D_\p = (Q, \P_0, \delta, q_\initial, F)$   where $$\delta = \{(q, b\wedge b', q') \mid (q, b', q')\in\delta_1 \}.$$

$D_\p$ is deterministic because $D_1$ is deterministic.

\paragraph{Sequencing (\rm $\p_1; \p_2$).}

Let the DFA for $\p_1$ and $\p_2$ be $D_1 = (Q_1, \P_0, \delta_1, q^1_\initial, F_1)$ and $D_2 = (Q_2, \P_0, \delta_2, q^2_\initial, F_2)$, respectively.  
Construct the NFA $N_\p = ( Q_1 \sqcup Q_2, \P_0, \delta, q^1_\initial, F_2)$ where $$\delta = \delta_1 \cup \delta_2 \cup \bigcup_{f\in F_1} \mathsf{DivertAwayFrom}(f)$$
where $\mathsf{DivertAwayFrom}(f)
= \{ (q_1^1, b, q_\initial^2)\} ~|~ (q_1^1, b, f) \in \delta_1\}$. 
Essentially, transitions in $\mathsf{DivertAwayFrom}(f)$ divert all incoming transitions to $f\in F_1$ to the initial state of the second DFA. 


Then, the DFA $D_\p$ is obtained from determinization of $N_\p$.

\paragraph{Choice (\rm $\p::=\choice{\p_1}{\p_2}$).}

Let the DFA for $\p_1$ and $\p_2$ be $D_1 = (Q_1, \P_0, \delta_1, q^1_\initial, F_1)$ and $D_2 = (Q_2, \P_0, \delta_2, q^2_\initial, F_2)$, respectively.  
Construct NFA $N_\p = (Q_1 \sqcup Q_2 \setminus \{q_\initial^2\}, \P_0, \delta, q_\initial^1, F_1\sqcup F_2)$ where
\begin{align*}
    \delta =  ~\delta_1 \ 
    \cup &~ \delta_2 \setminus \{(q^2_1,b, q^2_2) ~|~ (q^2_1,b, q^2_2)\in \delta_2 \text{ and } q^2_1 = q^2_\initial \} \\
    \cup &~ \{(q^1_\initial, b, q^2_2) ~ | ~ (q^2_1, b, q^2_2)  \in \delta_2 \text{ and } q^2_1 = q^2_\initial\}
\end{align*}

Then, the DFA $D_\p$ is obtained from determinization of $N_\p$.

Lastly, we can extend DFA $D_\p$ to make it {\em complete}, i.e., if for all states $q\in Q$ and assignments $\sigma \in 2^{\P_0}$,   $|\{q' \mid (q, b, q') \in \delta \text{ and } \sigma \models b  \}| = 1$.
\end{proof}

\paragraph{Reward Machine.}

Given a $\spectrl$ specification $\p$, let $D_\p = (Q, \P_0, \delta, q_\initial, F)$ be the DFA such that a trajectory $\zeta \models \p$ iff $L(\zeta)$ is accepted by the DFA $D_\p$, where $L: \S\to2^{\P_0}$ is the labelling function. WLOG, assume $\D_\p$ is complete. 

Construct a reward machine $\Rc_\p = (Q\cup \{\dead\}, \delta_u, \delta_r, q_\initial) $ where
the state transition function $\delta_u: \S\times \A\times Q \to Q$ is defined as
\begin{align*}
    \delta_u(s, \--, q) = 
    q' & \text{where } (q, b, q') \in \delta \text{ s.t. } L(s) \models  b
\end{align*}
and the reward function $\delta_r:\S\times Q \to [-1, 1]$ is given by 
\begin{align*}
    \delta_r(s, q) = 
    \begin{cases}
    1 & \text{if } q\notin F, q'  = \delta_u(s, \--, q')  \text{ and } q'\in F\\
    -1 & \text{if } q\in F, q'  = \delta_u(s, \--, q')  \text{ and } q'\notin F\\
    0 & \text{otherwise }
    \end{cases}
\end{align*}
Observe that the above functions are well defined since the DFA is deterministic and complete.

\begin{proof}[Proof of Theorem~\ref{thm:rm_const}]
Let $\Rc_\p$ be as constructed above. 
Let $ \zeta = s_0,s_1, \dots s_t,s_{t+1}$.
Then, by construction a run $\rho = q_0, q_1\dots q_{t+1}$ of $L(\zeta_{0:t})$ in $D_\p$ is also a run of $\zeta$ in $\Rc_\p$. Then, the reward function is design so that (a) each time the run visits a state in $F$ from a non-accepting state, it will receive a reward of 1, (b) each time the run visits a state in $Q\setminus F$ from a state in $F$, and it receives -1, and (c) 0 otherwise. 

Suppose $\zeta_{0:t}$ does not satisfy $\p$. Then $\rho$ is not an accepting run in $D_\p$. Then, each time the run  visits a state in $ F$, the run will exit states in $F$ after a finite amount of time. Thus, either $\zeta$ receives a reward of $0$ or it receives  a reward of $1$ and $-1$ an equal number of times. In this case,  $\Rc_\p(\zeta) = 0$ since the $+1$s and $-1$s will cancel each other out. 

Suppose $\zeta \models \p$.
Then, $\rho$ is an accepting run in $D_\p$. Let $k$ be the largest index such that $ q_k \notin F$ and $q_\ell \in F$ for all $k< \ell \leq t+1$. So, $\delta_r(s_k,q_k) = 1$ and $\delta_r(s_\ell,q_\ell) = 0$ for all $k<\ell\leq t$.
Additionally, the run $q_0, \dots, q_{k}$ is not an accepting run in $\D_\p$. Thus, the trajectory $\zeta_{0:k-1} $ does not satisfy $\p$, thus $\Rc_\p(\zeta_{0:k-1}) = 0$. So, $\Rc_\p(\zeta) = \Sigma_{z=0}^{t} \delta_r(s_z, q_z) = \Sigma_{z=0}^{k-1} \delta_r(s_z, q_z) + \delta_r(s_k,q_k) + \Sigma_{z=k+1}^{t} \delta_r(s_z,q_z)$. Since $\Sigma_{z=0}^{k-1} \delta_r(s_z, q_z) = \Sigma_{z=k+1}^{t} \delta_r(s_z,q_z) = 0$, we get that $\Rc_\p(\zeta)=1$.\qed
\end{proof}

\subsection{Simulating Punishment Game}\label{app:puni_game}

The optimal deviation score of agent $j$ w.r.t $\pi$, $\code{dev}_j^{\pi}$, is the min-max value of a two player zero-sum Markov game in
which the max-agent is agent $j$ of $\M$ and the min-agent is a coalition of punishing agents $\{i\in[n]\mid i\neq j\}$ whose combined policy is constrained to be of the form $(\pi\Join\tau)_{-j}$. Before solving this min-max game, we eliminate the constraint on the min-agent's policy by constructing a product of $\M$ with reward machine $\Rc_j = (Q_j, \delta_u^j, \delta_r^j, q_0^j)$ and finite-state deterministic joint policy $\pi = (M, \alpha, \sigma, m_0)$. We now describe the product construction.

\begin{proof}[Proof of Theorem~\ref{thm:pun_game}]
For an agent $j$, the two-player zero-sum game $\M_j$ is defined by $\M_j = (\S_j, A_j, \A_{-j}, P_j, H+1, s_0^j, R_j)$ with rewards where,
\begin{itemize}
    \item The set of states $\S_j$ is the product $\S_j = \S\times M\times Q_j\times\{\bot,\top\}$.
    \item The set of actions of the max-agent is $A_j$ and the set of actions of the min-agent is $\A_{-j} = \prod_{i\neq j}A_i$. Given $a_j\in A_j$ and $a_{-j}\in\A_{-j}$, we denote the joint action by $(a_j,a_{-j})\in\A$.
    \item The last component of a state denotes whether agent $j$ has deviated from $\pi_j$ in the past or not. Intuitively, $\bot$ implies that agent $j$ has \emph{not} deviated from $\pi_j$ in the past and $\top$ implies that it has deviated from $\pi_j$ in the past. We define an update function $f_j$ which is used to update this information at every step. The deviation update function $f_j:\S\times\A\times M\times \{\bot,\top\}\to\{\bot,\top\}$ is defined by $f_j(s,a,m,\top) = \top$ and $f_j(s,a,m,\bot) = \bot$ if $a_j = \sigma(s,m)_j$ and $\top$ otherwise.  \\
    
    The transitions of $\M_j$ are such that the action of the min-agent $a_{-j}$ is ignored and replaced with the output of $\pi_{-j}$ until agent $j$ deviates from $\pi_j$ (or equivalently, until the last component of the state is $\top$). The transition probabilities are given by
    \begin{itemize}
        \item $P_j((s,m,q,b), a, (s',m',q',b')) = P(s,(a_j,\sigma(s,m)_{-j}),s')$ if $m' = \alpha(s,(a_i,\sigma(s,m)_{-j}),m)$, $q' = \delta_u(s,(a_i,\sigma(s,m)_{-j}),q)$, $b = \bot$ and $b' = f_j(s,a,m,b)$.
        \item $P_j((s,m,q,b), a, (s',m',q',b')) = P(s,a,s')$ if $m' = \alpha(s,a,m)$, $q' = \delta_u(s,a,q)$, $b = \top$ and $b' = f_j(s,a,m_j,b)$.
        \item $P_j((s,m,q,b), a, (s',m',q',b')) = 0$ otherwise.
    \end{itemize}
    
    \item The initial state is $s_0^j = (s_0,m_0, q_0^j,\bot)$.
    \item The rewards are given by $R_j((s,m,q,b), a) = \delta_r^j(s, q)$.
\end{itemize}
Let us denote by $\bar{\pi}_1$ and $\bar{\pi}_2$ the policies of the max-agent and the min-agent respectively. Then the expected reward attained by the max-agent is
$$\bar{J}_j(\bar{\pi}_1, \bar{\pi}_{2}) = \E\Big[\sum_{k=0}^{H} R_j(\bar{s}_k,a_k)\mid \bar{\pi}_1,\bar{\pi}_2\Big]$$
where the expectation is w.r.t. the distribution over trajectories of length $H+1$ generated by using $(\bar{\pi}_1,\bar{\pi}_{2})$ in $\M_j$. Given a trajectory $\bar{\zeta} = \bar{s_0}\xrightarrow{a_0}\bar{s_1}\xrightarrow{a_1}\ldots \xrightarrow{a_{t-1}}\bar{s_t}$ in $\M_j$, we denote by $\bar{\zeta}\downarrow_{\M}$ the trajectory projected to the state space of $\M$---i.e., $\bar{\zeta}\downarrow_{\M} = {s_0}\xrightarrow{a_0}{s_1}\xrightarrow{a_1}\ldots \xrightarrow{a_{t-1}}{s_t}$. From Theorem~\ref{thm:rm_const} and the above definition of $\M_j$ it follows that for any trajectory $\bar{\zeta}$ in $\M_j$ of length $H+1$ we have $$\sum_{k=0}^{H} R_j(\bar{s}_k,a_k) = \Rc_j(\bar{\zeta}\downarrow_{\M}) = \mathbbm{1}(\bar{\zeta}_{0:H}\downarrow_{\M}\models \p_j).$$
Let $\D^{\M}[\pi]$ denote the distribution over length $H$ trajectories in $\M$ generated by $\pi$ and $\D^{\M_j}[\bar{\pi}]$ denote the distribution over length $H+1$ trajectories in $\M_j$ generated by $\bar{\pi}$.
It is easy to see that any policy $\pi_j'$ for agent $j$ in $\M$ can be interpreted as a policy $g(\pi_j')$ for the max-agent in $\M_j$ and any policy $\bar{\pi}_1$ for the max-agent in $\M_j$ can be interpreted as a policy $g'(\bar{\pi}_1)$ for agent $j$ in $\M$. Since the actions of the min-agent in $\M_j$ are only taken into account after agent $j$ deviates from $\pi_j$, we also have that any policy of the form $(\pi\Join\tau)_{-j}$ for the punishing agents in $\M$ corresponds to a policy $h(\tau)$ of the min-agent in $\M_j$ and any policy $\bar{\pi}_2$ of the min-agent in $\M_j$ corresponds to a policy $(\pi\Join h'(\bar{\pi}_2))_{-j}$ for the punishing agents in $\M$. Furthermore the mappings $g,g',h,h'$ satisfy the property that for any trajectory $\bar{\zeta}$ of length $H+1$ in $\M_j$, we have
\begin{itemize}
    \item for any $\tau$ and $\pi_j'$ in $\M$, $\D^{\M}[(\pi\Join\tau)_{-j}, \pi_j'](\bar{\zeta}\downarrow_{\M}) = \D^{\M_j}[g(\pi_j'), h(\tau)](\bar{\zeta})$ and,
    \item for any $\bar{\pi}_1$ and $\bar{\pi}_2$ in $\M$, $\D^{\M}[(\pi\Join h'(\bar{\pi}_2))_{-j}, g'(\bar{\pi}_1)](\bar{\zeta}\downarrow_{\M}) = \D^{\M_j}[\bar{\pi}_1, \bar{\pi}_2](\bar{\zeta})$. 
\end{itemize}

Therefore

\begin{align*}
\code{dev}_j^{\pi} &= \min_{\tau[j]}\max_{\pi_j'} J_j((\pi\Join\tau)_{-j}, \pi_j') \\
&= \min_{\tau[j]}\max_{\pi_j'}\E_{\zeta\sim\D^{\M}[(\pi\Join\tau)_{-j}, \pi_j']}\Big[\zeta\models\p_i\Big]\\
&= \min_{\bar{\pi}_2}\max_{\bar{\pi}_1}\E_{\bar{\zeta}\sim\D^{\M_j}[\bar{\pi}_1,\bar{\pi}_2]}\Big[\sum_{k=0}^{H} R_j(\bar{s}_k,a_k)\Big]\\
&= \min_{\bar{\pi}_2}\max_{\bar{\pi}_1}\bar{J}_{j}(\bar{\pi}_1,\bar{\pi}_2).
\end{align*}
It is easy to see that the function $h'$ has the desired properties of $\textsc{PunStrat}$. Finally, we observe that given a simulator for $\M$ it is straightforward to construct a simulator for $\M_j$. Similarly, given an estimate $\tilde{\M}$ of $\M$ we can use the above definition of $\M_j$ to construct an estimate $\tilde{\M}_j$ of $\M_j$.\qed
\end{proof}



\subsection{Solving Punishment Games}\label{app:bfsestimate}

Our sample efficient algorithm for solving the punishment game relies on Assumption~\ref{assump:model}. There are algorithms for solving min-max games (with unknown transition probabilities) without this assumption \cite{pmlr-v119-bai20a} albeit with slightly worse sample complexity. We outline the details of our algorithm below.

\begin{algorithm}[t]
\begin{algorithmic}[1]
\STATE $\tilde{M} \leftarrow (\tilde{\S}=\{s_0\}, \A, \tilde{P}=\emptyset, H, s_0)$
\STATE $K\leftarrow \Big\lceil \frac{2|\S|^2|M|^2|Q|^2H^4}{\delta^2}\log\Big(\frac{2|\S|^2|\A|}{p}\Big) \Big\rceil$
\STATE $\code{queue}\leftarrow$ $[s_0]$
\WHILE{$\lnot\ \code{queue.isempty()}$}
\STATE $s\leftarrow\code{queue.pop()}$
\FOR{$a\in\A$}
\STATE {\color{blue}// Initialize number of visits to each state}
\STATE $N\leftarrow\code{empty-map()}$
\FOR{$s'\in\tilde{\S}$}
\STATE $N[s'] \leftarrow 0$
\ENDFOR
\STATE {\color{blue}// Obtain $K$ samples for the state-action pair $(s,a)$}
\FOR{$x\in\{1,\ldots,K\}$}
\STATE $s'\sim P(\cdot \mid s,a)$
\STATE {\color{blue}// Add any newly discovered state to $\tilde{\S}$ and the map $N$}
\IF{$s'\notin \tilde{\S}$}
\STATE $\tilde{\S}\leftarrow\tilde{\S}\cup\{s'\}$
\STATE $\code{queue.add}(s')$
\STATE $N[s']\leftarrow 0$
\ENDIF
\STATE {\color{blue}// Increment number of visits to $s'$}
\STATE $N[s'] \leftarrow N[s'] + 1$
\ENDFOR
\STATE {\color{blue}// Store estimated transition probabilities in $\tilde{P}$}
\FOR{$s'\in \tilde{S}$}
\STATE $\tilde{P}(s'\mid s,a) \leftarrow \frac{N[s']}{K}$
\ENDFOR
\ENDFOR
\ENDWHILE
\STATE \textbf{return} $\tilde{\M}$
\caption{\textsc{BFS-Estimate}\\
Inputs: Precision $\delta$, failure probability $p$.\\
Outputs: Estimated model $\tilde{\M}$ of $\M$.}
\label{alg:estimate}
\end{algorithmic}
\end{algorithm}

\paragraph{Estimating $\M$ under Assumption~\ref{assump:model}.} The first step is to estimate the transition probabilities of $\M$ using \textsc{BFS-Estimate} which is outlined in Algorithm~\ref{alg:estimate}. \textsc{BFS-Estimate} performs a breadth-first-search on the transition graph of $\M$. In order to figure out all outgoing edges from a state $s$, multiple samples are collected by taking each possible action $K$-times from $s$. Newly discovered states are then added to the state space of $\tilde{\M}$ and the collected samples are used to estimate transition probabilities. The value $K$ is defined in line 2 of the algorithm in which $|Q|$ denotes the maximum size of the state space of reward machine $\Rc_j$ for any agent $j$---i.e., $|Q| = \max_{j}|Q_j|$. \textsc{BFS-Estimate} has the following approximation guarantee.

\begin{lemma}
With probability at least $1-p$, for all $s\in\tilde{\S}$, $a\in\A$ and $s'\in\S$,
$$\big\lvert\tilde{P}(s'\mid s,a) - P(s'\mid s,a)\big\rvert \leq \varepsilon = \frac{\delta}{2|\S||M||Q|H^2}$$
where $\tilde{P}(s'\mid s,a)$ is taken to be $0$ if $s'\notin \tilde{\S}$.
\end{lemma}\label{lem:prob_estimates}
\begin{proof}
For any given $s\in\tilde{\S}$, $a\in\A$ and $s'\in\S$, the probability $\tilde{P}(s'\mid s, a)$ is estimated using $K$ independent samples from $P(\cdot\mid s,a)$. Therefore, using Chernoff bounds, we get
$$\Pr\Big[\big\lvert\tilde{P}(s'\mid s,a) - P(s'\mid s,a)\big\rvert > \varepsilon\Big] \leq 2e^{-2K\varepsilon^2}$$
Applying union bound over all triples $(s,a,s')\in\tilde{\S}\times\A\times\S$ and substituting the values of $K$ and $\varepsilon$, we get
\begin{align*}
    \Pr\Bigg[\bigcup_{s,a,s'}\Big\{\big\lvert\tilde{P}(s'\mid s,a) - P(s'\mid s,a)\big\rvert > \varepsilon\Big\}\Bigg] &\leq 2|\S|^2|\A|e^{-2K\varepsilon^2}\\
    &\leq 2|\S|^2|\A|e^{-\log(\frac{2|\S|^2|\A|}{p})}
    = p.
\end{align*}
Hence we obtained the desired bound.\qed
\end{proof}

\paragraph{Obtaining estimates of $\M_j^{\pi}$.} After estimating $\M$, for any finite-state deterministic joint policy $\pi$ and any agent $j$, we perform the product construction outlined in Section~\ref{app:puni_game} with the estimated model $\tilde{\M}$ to obtain an estimate $\tilde{\M}_j^{\pi}$ of the punishment game $\M_j^{\pi}$. The constructed model $\tilde{\M}_j^{\pi}$ can be used to estimate $\code{dev}_j^{\pi}$ as claimed in Theorem~\ref{thm:estimate}.

\begin{proof}[Proof of Theorem~\ref{thm:estimate}]
Since the transition probabilities in $\tilde{\M}_j^{\pi}$ are inherited from $\tilde{\M}$, from Lemma~\ref{lem:prob_estimates} we have that with probability at least $1-p$, for any $\pi$, $j$, $\bar{s},\bar{s}'\in\S_j$ and $a\in\A$ such that $\bar{s}$ is in the state space of $\tilde{\M}_j^{\pi}$, 
$$\big\lvert\tilde{P}_j(\bar{s}'\mid \bar{s},a) - P_j(\bar{s}'\mid \bar{s},a)\big\rvert \leq \varepsilon = \frac{\delta}{2|\S||M||Q|H^2}$$
where $\tilde{P}_j$ represents the transition probabilities of $\tilde{\M}_j$\footnote{Omitting the superscript $\pi$ in $\tilde{\M}_j^{\pi}$}. Consider the model $\tilde{\M}_j'$ which is a modification of $\tilde{\M}_j$ that has state space $\S_j$ (state space of $\M_j$) and for any $\bar{s},\bar{s}'\in\S_j$ and $a\in\A$, its transition probabilities are defined by $\tilde{P}_j'(\bar{s'}\mid\bar{s},a) = \tilde{P}_j(\bar{s'}\mid\bar{s},a)$ if $\bar{s}$ is in the state space of $\tilde{M}_j$ and $P_j(\bar{s'}\mid\bar{s},a)$ otherwise. We have $\big\lvert\tilde{P}_j'(\bar{s}'\mid \bar{s},a) - P_j(\bar{s}'\mid \bar{s},a)\big\rvert \leq \varepsilon$ for all $\bar{s},\bar{s}'\in\S_j$ and $a\in\A$. For any two policies $\bar{\pi}_1$ and $\bar{\pi}_2$ of the max-agent and the min-agent in $\M_j$ respectively, we denote by $\bar{J}^{\tilde{\M}_j}(\bar{\pi}_1,\bar{\pi}_2)$ and $\bar{J}^{\tilde{\M}_j'}(\bar{\pi}_1,\bar{\pi}_2)$ the expected reward over $H+1$ length trajectories generated by $(\bar{\pi}_1,\bar{\pi}_2)$ in $\tilde{\M}_j$ and $\tilde{\M}_j'$ respectively. Then $\bar{J}^{\tilde{\M}_j}(\bar{\pi}_1,\bar{\pi}_2) = \bar{J}^{\tilde{\M}_j'}(\bar{\pi}_1,\bar{\pi}_2)$ because both the models assign the same probability to all runs as any run that leaves the state space of $\tilde{\M}_j$ has probability zero in both the models. Now we can apply Lemma 4 of \cite{rmax} to conclude that $|\bar{J}^{\tilde{\M}_j'}(\bar{\pi}_1,\bar{\pi}_2) - \bar{J}_j^{\pi}(\bar{\pi}_1,\bar{\pi}_2)| \leq \delta$ and hence $|\bar{J}^{\tilde{\M}_j}(\bar{\pi}_1,\bar{\pi}_2) - \bar{J}_j^{\pi}(\bar{\pi}_1,\bar{\pi}_2)| \leq \delta$ for any $\bar{\pi}_1$ and $\bar{\pi}_2$. This implies that for any $\bar{\pi}_2$ we have
\begin{equation}\label{eq:intresult}
    |\max_{\bar{\pi}_1}\bar{J}^{\tilde{\M}_j}(\bar{\pi}_1,\bar{\pi}_2) - \max_{\bar{\pi}_1}\bar{J}_j^{\pi}(\bar{\pi}_1,\bar{\pi}_2)|\leq \delta
\end{equation}
and therefore can conclude that
$$|\min_{\bar{\pi}_2}\max_{\bar{\pi}_1}\bar{J}^{\tilde{\M}_j}(\bar{\pi}_1,\bar{\pi}_2) - \min_{\bar{\pi}_2}\max_{\bar{\pi}_1}\bar{J}_j^{\pi}(\bar{\pi}_1,\bar{\pi}_2)|\leq \delta.$$
Applying Theorem~\ref{thm:pun_game} we get
$$|\min_{\bar{\pi}_2}\max_{\bar{\pi}_1}\bar{J}^{\tilde{\M}_j}(\bar{\pi}_1,\bar{\pi}_2) - \code{dev}_j^{\pi}|\leq \delta.$$
Now, let $\bar{\pi}_2^* = \arg\min_{\bar{\pi}_2}\max_{\bar{\pi_1}}\bar{J}^{\tilde{\M}_j}(\bar{\pi}_1,\bar{\pi}_2)$ and $\tau[j] = \textsc{PunStrat}(\bar{\pi}_2^*)$. Then from Equation~\ref{eq:intresult} we can conclude that $|\max_{\bar{\pi}_1}\bar{J}^{\tilde{\M}_j}(\bar{\pi}_1,\bar{\pi}_2^*) - \max_{\bar{\pi}_1}\bar{J}_j^{\pi}(\bar{\pi}_1,\bar{\pi}_2^*)|\leq \delta$. Using Theorem~\ref{thm:pun_game} we get
$$\Big|\max_{\bar{\pi}_1}\bar{J}^{\tilde{\M}_j^{\pi}}(\bar{\pi}_1, \bar{\pi}_{2}^*) - \max_{\pi_j'}J_j((\pi\Join\tau[j])_{-j},\pi_j')\Big| \leq \delta.$$
Finally the total number of samples used is at most $|\S||\A|K = O\left(\frac{|\S|^3|M|^2|Q|^4|\A|H^4}{\delta^2}\log\left(\frac{|\S||\A|}{p}\right)\right)$.\qed
\end{proof}
For each $j$, the min-max game $\tilde{\M}_j$ is solved in polynomial time using value iteration \cite{pmlr-v119-bai20a} to compute an estimate $\tilde{\code{dev}}_j$ of $\code{dev}_j^{\pi}$ which is used in Line~\ref{algline:devcheck} of Algorithm~\ref{alg:NEverification} to check whether agent $j$ can successfully deviate from $\pi_j$.

\subsection{Soundness Guarantee}\label{app:soundness} The soundness guarantee of \textsc{HighNashSearch} follows from Theorem~\ref{thm:estimate}.

\begin{proof}[Proof of Corollary~\ref{cor:soundness}]
From Theorem~\ref{thm:estimate} we get that with probability at least $1-p$, if a policy $\tilde{\pi}=\pi\Join\tau$ is returned by our algorithm, then for all $j$ we have 
\begin{align*}
    \max_{\pi_j'}J_j((\pi\Join\tau)_{-j},\pi_j')&\leq\tilde{\code{dev}}_j+\delta&&\text{(Equation~\ref{eq:puneffective})}\\
    &\leq J_j(\pi)+\epsilon&&\text{(Line~\ref{algline:devcheck} of Algorithm~\ref{alg:NEverification})}\\
    &=J_j(\pi\Join\tau) + \epsilon.
\end{align*}
Therefore, $\pi\Join\tau$ is an $\epsilon$-Nash equilibrium.\qed
\end{proof}

\section{Baselines}\label{app:baselines}
\subsection{\textsc{Nash Value Iteration}} 
\begin{algorithm}[t]
\begin{algorithmic}[1]
\STATE Initialize joint policy $\pi = (\pi_1,\ldots,\pi_n)$
\STATE Initialize value function $V: \S\times[H+1]\to\R^n$ to be the zero map
\FOR{$t\in\{H,H-1,\ldots,1\}$}
\FOR{$s\in\S$}
\STATE Initialize step game $G_{s}^t:\A\to\R^n$
\FOR{$a = (a_1,\ldots,a_n)\in\A$}
\STATE $G_{s}^t(a_1,\ldots,a_n) = R(s,a) + \E_{s'\sim P(\cdot\mid s,a)}[V(s', t+1)]$
\ENDFOR
\STATE $(d_1,d_2,\ldots,d_n) \gets\textsc{Best-Nash-General-Sum}(G_s^t)\in\D(A_1)\times\cdots\times\D(A_n)$
\STATE $V(s,t)\gets\E_{a_1\sim d_1, a_2\sim d_2,\ldots, a_n\sim d_n}[G_s^t(a_1,\ldots,a_n)]$
\STATE $\pi(s,t)\gets(d_1,d_2,\ldots,d_n)$
\ENDFOR
\ENDFOR
\STATE \textbf{return} $\pi$
\caption{Nash Value Iteration\\
Inputs: $n$-agent Markov game $\M$ with rewards, horizon $H$.\\
Outputs: Nash equilibrium joint policy $\pi = (\pi_1,\ldots,\pi_n)$.}
\label{alg:nvi}
\end{algorithmic}
\end{algorithm}

This baseline first computes an estimate $\tilde{\M}$ of $\M$ using \textsc{BFS-Estimate} (Algorithm~\ref{alg:estimate}) and then computes a product of $\tilde{\M}$ with the reward machines corresponding to the agent specifications in order to define rewards at every step. It then solves the resulting general sum game $\tilde{\M}'$ using value iteration. The value iteration procedure is outlined in Algorithm~\ref{alg:nvi} which uses \textsc{Best-Nash-General-Sum} to solve $n$-player general-sum strategic games (one-step games) at each step. When there are multiple Nash equilibria for a step game, \textsc{Best-Nash-General-Sum} chooses one with the highest social welfare (for that step). In our experiments, we use the library $\code{gambit}$ \cite{gambit} for solving the step games.

\subsection{\textsc{Multi-agent QRM}}

The second baseline (Algorithm~\ref{alg:maqrm}) is a multi-agent variant of QRM~\cite{icarte2018using,tor-etal-arxiv20}. We derive reward machines from agent specifications using the procedure described in Appendix~\ref{appendix:rm}.  

We learn one Q-function for each agent. The Q-function for the $i$-th agent, denoted $Q_i : S \times \Pi_{i\in [n]} U_i \to \A_i$, can be used to derive the best action for the $i$-th agent from the current state of the environment and reward machines of all agents. In every step, $Q_i$ is used to sample an action $a_i$ for the $i$-th agent. The joint action $(a_i)_{i\in [n]}$ is used to take a step in the environment and all reward machines. Finally, each $Q_i$ is individually updated according to the reward gained by the $i$-th agent. For notational convenience, we let $q \xrightarrow{\alpha} q'$ denote $q \gets (1-\alpha)\cdot q + \alpha\cdot q'$.

\begin{algorithm}[t]
\begin{algorithmic}[1]

\STATE{\textbf{for} $i \in [n]$ \textbf{do} $(U_i, \delta_u^i, \delta_r^i, u_0^i) \gets \mathsf{RewardMachine}(\p_i)$}\label{algline:learn:ConstructRewardMachine}

\STATE{\text{\color{blue} // Initialize environment state, reward machines state, and Q-functions}}
\STATE{Current state $s \gets s_0$}
\STATE {\textbf{for} $i \in [n]$ \textbf{do} $u_i \gets u^i_0$} \\
    
\STATE {\textbf{for} $i \in [n]$ \textbf{do} Initialize $Q_i(s,(u_1, \dots, u_n),a_i)$ for all states $s\in S, u_i \in U_i$, and actions $a_i \in\A_i $}

\FOR{$l \in \{0,\ldots,N\}$}
    \STATE{\color{blue} // Sample actions from policy derived from Q-functions}
    \STATE{\textbf{for} $i\in [n]$ \textbf{do} choose action $a_i\in \A_i$ at $(s, (u_1,\dots, u_n))$ using exploration policy derived from $Q_i$ (e.g., $\varepsilon$-greedy)} \\
    \STATE{\color{blue} // Take a step in environment and the reward machines} \\
    \STATE{Take action $a = (a_1, \dots a_n)$ in $\M$ and observe the next state $s'$} \\
    \STATE{\textbf{for} $i \in [n]$ \textbf{do} compute the reward $r_i \gets \delta_r^i(s, u_i)$ and next RM state $u'_i \gets \delta_u^i(s, a , u_i))$} \\
    \STATE{\color{blue} // Update all Q-functions} \\
   \IF{$s'$ is terminal}
        \STATE{\textbf{for} $i \in [n]$ \textbf{do} $Q_i(s, (u_1,\dots, u_n),a) \xleftarrow{\alpha} r_i $}
    \ELSE
        \STATE {\textbf{for} $i \in [n]$ \textbf{do} $Q_i(s, (u_1,\dots, u_n),a_i) \xleftarrow{\alpha} r_i + \gamma\cdot \max_{a'_i \in \A_i} Q_i(s', (u_1',\dots, u_n'), a'_i)$}
  \ENDIF   
   \IF{$s'$ is terminal}
        \STATE{\text{\color{blue} // Reset environment state and reward machines state}} \\
        \STATE{$s \gets s_0$ and \textbf{for} $i \in [n]$ \textbf{do} $u_i \gets u^i_0$}
    \ELSE
        \STATE {$s \gets s'$ and  \textbf{for} $i\in [n]$ \textbf{do} $u_i \gets u_i'$}
    \ENDIF
\ENDFOR
\STATE{\textbf{for} $i\in[n]$ \textbf{do} $\pi_i \gets$ Best action policy derived from $Q_i$} 
\STATE \textbf{return} $(\pi_1, \dots, \pi_n)$
\caption{Multi-agent QRM\\
Inputs: $n$-agent Markov game $\M = (\S, \A=\Pi_{i\in [n]} \A_i, P, H, s_0)$, agent specifications $\p_1, \dots \p_n$, learning rate $\alpha \in (0, 1]$, discount factor $\gamma \in (0, 1]$, $\varepsilon \in (0,1]$\\
Outputs: Joint policy $\pi = (\pi_1,\ldots,\pi_n)$.}
\label{alg:maqrm}
\end{algorithmic}
\end{algorithm}

\section{Benchmark Details}\label{sec:benchmarks}
\subsection{Intersection}
The specifications and the corresponding initial states for the intersection environment are described below.

\begin{description}
   \item[$\p^1$] Two N-S cars both starting at 3 and one E-W car starting at 2. N-S cars' goal is to reach 0 before the E-W car without collision. E-W car's goal is to reach 0 before both N-S agents without collision.
   
   \item[$\p^2$] Same as motivating example except that blue car is not required to stay a car length away from green and orange cars.
   
   \item[$\p^3$] Same as motivating example.
   
   \item[$\p^4$] Two N-S agents (0 and 1) both starting at 3 and one E-W agent (2) starting at 3. Agent 0's task is to reach 0 before other two agents. Agent 1's task is to reach 0. Agent 2's task is to reach 0 before agent 1. All agents must avoid collision.
   
   \item[$\p^5$] Two N-S cars starting at 2 and 3 and three E-W cars all starting at 2, 3 and 4, respectively. N-S cars' goal is to reach 0 before the E-W cars without collision. E-W cars' goal is to reach 0 before both N-S cars without collision.
\end{description}

\subsection{Single Lane Environment}

The environment consists of $k$ agents along a straight track of length $l$. All agents are initially placed at the $0^{\text{th}}$ location and the destination is at the $l^{\text{th}}$ location. In a single step, each agent can either move forward one location (with a failure probability of 0.05) or remain in its current position. We create competitive and co-operative scenarios through agent specifications. For example, we can create competitive scenarios in which an agent meets its specification only if it reaches the final location before all other or a set of other agents. We can also create co-operative scenarios in which the $i$-th agent must reach its goal before the $j$-th agent, for various pairs of agents 
($i,j$). In these cases, the highest social welfare would occur when the agents manage to coordinate so that all ordering constraints are satisfied (we ensure that there are no cycles in the ordering constraints). The specifications are described below and results are in Table~\ref{tab:sort_results}.

\paragraph{Specifications}

\begin{description}
    \item[$\p^1$] All agents should reach the final state

    
    \item [$\p^2$]  Agent 1 should reach its destination before Agent 0. Agent 2 must reach the destination. 
    

    \item [$\p^3$]  Agent 1 should reach its destination before Agent 0. All agents must reach their destination. 
    
    
    \item [$\p^4$] Agent 0 and Agent 1 are competing to reach the destination first. Only the agent reaching first meets the specification. Agent 2's goal is to reach the destination. 
    
    
    \item [$\p^5$] Agent 0 should reach the mid-point before Agent 1. However, Agent 1 should reach the final destination before Agent 0. All agents should eventually reach the final destination. 

    
    \item [$\p^6$] Agent 0 should reach the mid-point before Agent 1. However, Agent 1 and Agent 2 should reach the final destination before Agent 0. All agents should reach the final destination.  


\end{description}

\begin{table*}
\centering\scriptsize
\begin{tabular}{cccrrrr}
\toprule
\multirow{3}{*}{Spec.} &
\multirow{3}{*}{\begin{tabular}[c]{@{}c@{}}Num. of\\ agents\end{tabular}} &
\multirow{3}{*}{Algorithm} &
\multirow{3}{*}{\begin{tabular}[c]{@{}c@{}}\qquad$\code{welfare}(\pi)$\qquad\\ \qquad(avg \textpm ~ std)\qquad\end{tabular}} &
\multirow{3}{*}{\begin{tabular}[c]{@{}c@{}}\qquad$\epsilon_{\min}(\pi)$\qquad\\ \qquad(avg \textpm ~ std)\qquad\end{tabular}} &
\multirow{3}{*}{\begin{tabular}[c]{@{}c@{}}Num. of\\ runs terminated\end{tabular}} &
\multirow{3}{*}{\begin{tabular}[c]{@{}c@{}}Avg. num. of\\ sample steps\\(in millions)\end{tabular}} \\
{} & {} & {} & {} & {} & {}\\
{} & {} & {} & {} & {} & {}\\
\midrule
\multirow{3}{*}{$\p^1$} &
\multirow{3}{*}{3} &
\nash & 1.00 \textpm ~ 0.00 & 0.00 \textpm ~ 0.00 & 10 & 3.02 \\
{} & {} & \textsc{nvi} & 1.00 \textpm ~ 0.00 & 0.00 \textpm ~ 0.00 & 10 & 5.00 \\
{} & {} & \textsc{maqrm} & 1.00 \textpm ~ 0.00 & 0.01 \textpm ~ 0.00 & 10 & 4.00 \\

\midrule
\multirow{3}{*}{$\p^2$} &
\multirow{3}{*}{3} &
\nash & 1.00 \textpm ~ 0.00 & 0.00 \textpm ~ 0.00 & 10 & 3.01 \\
{} & {} & \textsc{nvi} & 1.00 \textpm ~ 0.00 & 0.00 \textpm ~ 0.00 & 10 & 5.00 \\
{} & {} & \textsc{maqrm} & 0.66 \textpm ~ 0.00 & 0.99 \textpm ~ 0.00 & 10 & 4.00
 \\

\midrule
\multirow{3}{*}{$\p^3$} &
\multirow{3}{*}{3} &
\nash & 1.00 \textpm ~ 0.00 & 0.00 \textpm ~ 0.00 & 10 & 3.50 \\
{} & {} & \textsc{nvi} & 1.00 \textpm ~ 0.00 & 0.00 \textpm ~ 0.00 & 8 & 5.00 \\
{} & {} & \textsc{maqrm} & 0.81 \textpm ~ 0.01 & 0.56 \textpm ~ 0.03 & 10 & 4.00
 \\

\midrule
\multirow{3}{*}{$\p^4$} &
\multirow{3}{*}{3} &
\nash & 0.67 \textpm ~ 0.00 & 0.00 \textpm ~ 0.00 & 10 & 3.03 \\
{} & {} & \textsc{nvi} & 0.33 \textpm ~ 0.00 & 0.00 \textpm ~ 0.00 & 10 & 5.00 \\
{} & {} & \textsc{maqrm} & 0.63 \textpm ~ 0.00 & 0.57 \textpm ~ 0.01 & 10 & 4.00
\\

\midrule
\multirow{3}{*}{$\p^5$} &
\multirow{3}{*}{3} &
\nash & 1.00 \textpm ~ 0.00 & 0.00 \textpm ~ 0.00 & 10 & 3.60 \\
{} & {} & \textsc{nvi} & 0.33 \textpm ~ 0.00 & 0.00 \textpm ~ 0.00 & 10 & 5.00 \\
{} & {} & \textsc{maqrm} & 0.62 \textpm ~ 0.01 & 0.47 \textpm ~ 0.02 & 10 & 4.00 \\

\midrule
\multirow{3}{*}{$\p^6$} &
\multirow{3}{*}{3} &
\nash & 1.00 \textpm ~ 0.00 & 0.00 \textpm ~ 0.00 & 10 & 3.68 \\
{} & {} & \textsc{nvi} & 0.00 \textpm ~ 0.00 & 0.00 \textpm ~ 0.00 & 10 & 5.00 \\
{} & {} & \textsc{maqrm} & 0.44 \textpm ~ 0.01 & 0.55 \textpm ~ 0.02 & 10 & 4.00 \\
\bottomrule\\

\end{tabular}
\caption{Results for all specifications in Single Lane Environment. Total of 10 runs per benchmark. Timeout = 24 hrs.}
\label{tab:sort_results}
\end{table*}

\subsection{Gridworld}

The environment is a $4\times 4$ discrete grid with 2-agents. The agents are initially placed at opposite corners of the grid. In every step, each agent can either move in one of  four cardinal directions (with a failure probability) or remain in position. The task of each agent is to visit a series of locations on the grid while ensuring no collision between the agents. The agents must learn to coordinate between themselves to accomplish their tasks. As an example, each agent's task is to visit any one of the other two corners in the grid. In this case, the agents must learn to choose and navigate to different corners to minimize their risk of collision. This specification can be increased in length (thus, in complexity) by sequencing visits to more locations on the grid. All scenarios are cooperative. The specifications are described below and results are in Table~\ref{tab:grid_results}.

\paragraph{Specifications}

\begin{description}
    \item[$\p^1$] Swap the positions of both agents without collision. 
    
    \item [$\p^2$]  Both agents should choose to visit one of the other two corners of the grid without collision.

    \item [$\p^3$] Append $\p^3$ with the specification to reach the corner that is diagonally opposite the initial position of the agent without collision.

    \item [$\p^4$] Append $\p^4$ with the agents swapping their positions without collision. 
    
    \item [$\p^5$] Append $\p^5$ with the agents swapping their positions again without collision. 
    
    
\end{description}

\begin{table*}
\centering\scriptsize
\begin{tabular}{cccrrrr}
\toprule
\multirow{3}{*}{Spec.} &
\multirow{3}{*}{\begin{tabular}[c]{@{}c@{}}Num. of\\ agents\end{tabular}} &
\multirow{3}{*}{Algorithm} &
\multirow{3}{*}{\begin{tabular}[c]{@{}c@{}}\qquad$\code{welfare}(\pi)$\qquad\\ \qquad(avg \textpm ~ std)\qquad\end{tabular}} &
\multirow{3}{*}{\begin{tabular}[c]{@{}c@{}}\qquad$\epsilon_{\min}(\pi)$\qquad\\ \qquad(avg \textpm ~ std)\qquad\end{tabular}} &
\multirow{3}{*}{\begin{tabular}[c]{@{}c@{}}Num. of\\ runs terminated\end{tabular}} &
\multirow{3}{*}{\begin{tabular}[c]{@{}c@{}}Avg. num. of\\ sample steps\\(in millions)\end{tabular}} \\
{} & {} & {} & {} & {} & {}\\
{} & {} & {} & {} & {} & {}\\
\midrule
\multirow{3}{*}{$\p^1$} &
\multirow{3}{*}{2} &
\nash & 0.95 \textpm ~ 0.02 & 0.00 \textpm ~ 0.00 & 10 & 18.86 \\
{} & {} & \textsc{nvi} & 1.00 \textpm ~ 0.00 & 0.00 \textpm ~ 0.00 & 10 & 22.40 \\
{} & {} & \textsc{maqrm} & Timeout & Timeout & 10 & Timeout \\


\midrule
\multirow{3}{*}{$\p^2$} &
\multirow{3}{*}{2} &
\nash & 0.99 \textpm ~ 0.01 & 0.00 \textpm ~ 0.00 & 10 & 29.74 \\
{} & {} & \textsc{nvi} & 1.00 \textpm ~ 0.00 & 0.00 \textpm ~ 0.00 & 8 & 38.40 \\
{} & {} & \textsc{maqrm} & Timeout & Timeout & 10 & Timeout \\

\midrule
\multirow{3}{*}{$\p^3$} &
\multirow{3}{*}{2} &
\nash & 0.98 \textpm ~ 0.02 & 0.00 \textpm ~ 0.00 & 10 & 48.40 \\
{} & {} & \textsc{nvi} & 1.00 \textpm ~ 0.00 & 0.00 \textpm ~ 0.00 & 4 & 57.60 \\
{} & {} & \textsc{maqrm} & Timeout & Timeout & 10 & Timeout \\

\midrule
\multirow{3}{*}{$\p^4$} &
\multirow{3}{*}{2} &
\nash & 0.94 \textpm ~ 0.02 & 0.00 \textpm ~ 0.00 & 10 & 65.91 \\
{} & {} & \textsc{nvi} & 0.94 \textpm ~ 0.01 & 0.00 \textpm ~ 0.00 & 7 & 92.80 \\
{} & {} & \textsc{maqrm} & Timeout & Timeout & 10 & Timeout \\

\midrule
\multirow{3}{*}{$\p^5$} &
\multirow{3}{*}{2} &
\nash & 0.86 \textpm ~ 0.05 & 0.00 \textpm ~ 0.00 & 10 & 87.05 \\
{} & {} & \textsc{nvi} &  0.13 \textpm ~ 0.01 & 0.00 \textpm ~ 0.00 & 10 & 128.00 \\
{} & {} & \textsc{maqrm} & Timeout & Timeout & 10 & Timeout \\

\bottomrule\\

\end{tabular}
\caption{Results for all specifications in Gridworld Environment. Total of 10 runs per benchmark. Timeout = 24 hrs.}
\label{tab:grid_results}
\end{table*}

\section{Experimental Details}\label{sec:add_results}

\paragraph{Hyperparameters.} In our implementation of \nash, we used Q-learning with $\varepsilon$-greedy exploration to learn edge policies  with $\varepsilon=0.15$, learning rate of $0.1$ and discount factor $\gamma=0.9$. In the verification phase, we used Nash factor $\epsilon=0.06$ and precision value $\delta=0.01$. In our verification algorithm, the failure probability $p$ is only used within \textsc{BFS-Estimate} to compute the number of samples $K$ to collect for each state-action pair. In the experiments, we directly set the value of $K$ to $1000$.

\paragraph{Platform.} Our implementation of \nash and the baselines is in Python3. All experiments were run on a 80-core machine with processor speed 1.2GHz and Ubuntu 18.04.


\end{document}